\RequirePackage{fix-cm}
\documentclass[twocolumn, draft, natbib, envcountsame]{svjour3}          
\smartqed  
\usepackage{graphicx}
\usepackage[misc]{ifsym}
\usepackage{todonotes}
\presetkeys{todonotes}{inline}{}
\usepackage{amsmath}
\usepackage{amssymb}
\usepackage{enumitem}
\usepackage{bbold}
\usepackage{mathtools}
\usepackage{algorithm}
\usepackage{algpseudocode}
\usepackage{tikz}
\usepackage{booktabs, multirow}
\usepackage{url}
\newcommand{\myModifier}[1]{#1'}
\newcommand{\mySet}{\widetilde{\mathcal{S}}}
\newcommand{\SchSec}{S^\prime}

\newcommand{\abs}[1]{\lvert#1\rvert}
\newcommand{\st}{\mathrel{}\middle|\mathrel{}}

\newcommand{\mytheoremname}{DefaultName}
\spnewtheorem*{mytheorem}{\mytheoremname}{\bf}{\it}

\newcommand{\grid}[4]{
	
	\draw [->] (0,-0.9*#3) -- (#1*#3 + 1.5*#3, -0.9*#3) node[anchor=#4] {time};
	
	\foreach \time in {0,...,#1}{
		\draw (\time*#3, -0.7*#3) -- (\time*#3, -1.1*#3) node[anchor=north] {\time};
	}
	\foreach \job in {1,...,#2}{
		\node[anchor=east] at (-0.3*#3, 0.5*#3 + #2*#3-\job*#3) {$J_{\job}$};
	}
}
\newcommand{\dueColor}{black!20}
 \journalname{Journal of Scheduling}
\begin{document}

\title{Scheduling a Proportionate Flow Shop of Batching Machines
}


\author{Christoph Hertrich\textsuperscript{1,2}        \and
		Christian Wei\ss\textsuperscript{3}        \and
        Heiner Ackermann\textsuperscript{3}      \and
        Sandy Heydrich\textsuperscript{3}        \and
        Sven O. Krumke\textsuperscript{1}
}

\authorrunning{C. Hertrich, C. Wei\ss, H. Ackermann, S. Heydrich, S. O. Krumke} 

\institute{
	\begin{itemize}
		\item[] Christoph Hertrich\\
		hertrich@math.tu-berlin.de
		\item[] Christian Wei\ss\\
		christian.weiss@itwm.fraunhofer.de
		\item[] Heiner Ackermann\\
		heiner.ackermann@itwm.fraunhofer.de
		\item[] Sandy Heydrich\\
		sandy.heydrich@itwm.fraunhofer.de
		\item[] Sven O. Krumke\\
		krumke@mathematik.uni-kl.de
	\end{itemize}
	\begin{itemize}
		\item[\textsuperscript{1}]
		Technische Universit\"at Kaiserslautern\\
		Department of Mathematics\\
		67663 Kaiserslautern, Germany
		\item[\textsuperscript{2}] 
		Technische Universit\"at Berlin\\
		Institute of Mathematics\\
		10623 Berlin, Germany
		\item[\textsuperscript{3}] 
		Fraunhofer Institute for Industrial Mathematics ITWM\\
		Department of Optimization\\
		67663 Kaiserslautern, Germany
	\end{itemize}
	This is a post-peer-review, pre-copyedit version of an article published in the Journal of Scheduling. The final authenticated version is available online at: \url{https://doi.org/10.1007/s10951-020-00667-2}
}

\date{Received: date / Accepted: date}

\maketitle

\begin{abstract}
In this paper we study a proportionate flow shop of batching machines with release dates and a fixed number $m \geq 2$ of machines. 
The scheduling problem has so far barely received any attention in the literature, but recently its importance has increased significantly, due to applications in the industrial scaling of modern bio-medicine production processes. 
We show that for any fixed number of machines, the makespan and the sum of completion times can be minimized in polynomial time.
Furthermore, we show that the obtained algorithm can also be used to minimize the weighted total completion time, maximum lateness, total tardiness and (weighted) number of late jobs in polynomial time if all release dates are $0$.
Previously, polynomial time algorithms have only been known for two machines.
\end{abstract}

\section{Introduction \label{Sec:Introduction}}

Modern medicine can treat some serious illnesses using individualized drugs, which are produced to order for a specific patient and adapted to work only for that unique patient and nobody else. Manufacturing such a drug often involves a complex production line, consisting of many different steps.

If each step of the process is performed manually, by a laboratory worker, then each laboratory worker can only handle materials for one patient at a time.
However, in the industrial application which motivates our study, some steps are instead performed by actual machines, like pipetting robots.
Such high-end machines can typically handle materials for multiple patients in one go. 
If scheduled efficiently, this special feature can drastically increase the throughput of the production line.
Clearly, in such an environment efficient operative planning  is crucial in order to optimize the performance of the manufacturing process and treat as many patients as possible.

In this paper, we present -- as part of a larger study of an industrial application -- new theoretical results for scheduling problems arising from a manufacturing process as described above.
Formally, the manufacturing process we consider is structured in a \emph{flow shop} manner, where each step is handled by a single, dedicated machine. 
A \emph{job} $J_j$, $j=1,2, \ldots, n$, representing the production of a drug for a specific patient, has to be processed by \emph{machines} $M_1,M_2, \ldots, M_m$ in order of their numbering.
Each job $J_j$ has a \emph{release date} $r_j \geq 0$, denoting the time at which the job $J_j$ is available for processing at the first machine $M_1$.
Furthermore, a job is only available for processing at machine $M_i$, $i=2,3, \ldots, m$, when it has finished processing on the previous machine $M_{i-1}$. 
 
Processing times are job-independent, meaning that each machine $M_i$, $i=1,2, \ldots, m$, has a fixed \emph{processing time} $p_i$, which is the same for every job when processed on that machine.
In the literature, a flow shop with machine- or job-independent processing times is sometimes called a \emph{proportionate} flow shop, see e.g.\ \citep{Panwalkar:ProportionateReview}.

Recall that, as a special feature from our application, each machine in the flow shop can handle multiple jobs at the same time. 
These kind of machines are called \emph{(parallel) batching machines} and a set of jobs processed at the same time on some machine is called a \emph{batch} on that machine \citep[Chapter 8]{brucker:scheduling}.
All jobs in one batch $B^{(i)}_k$ on some machine $M_i$ have to start processing on $M_i$ at the same time.
In particular, all jobs in $B^{(i)}_k$ have to be available for processing on $M_i$, before $B^{(i)}_k$ can be started.
The processing time of a batch on $M_i$ remains $p_i$, no matter how many jobs are included in this batch. This distinguishes parallel batching machines from \emph{serial batching machines}, where the processing time of a batch is calculated as the sum of the processing times of the jobs contained in it plus an additional setup time.
Each machine $M_i$, $i=1,2, \ldots, m$ has a \emph{maximum batch size} $b_i$, which is the maximum number of jobs a batch on machine $M_i$ may contain.

Given a feasible schedule $S$, we denote by $c_{ij}(S)$ the \emph{completion time} of job $J_j$ on machine $M_i$. For the completion time of job $J_j$ on the last machine we also write $C_j(S)=c_{mj}(S)$. If there is no confusion which schedule is considered, we may omit the reference to the schedule and simply write $c_{ij}$ and $C_j$.

As optimization criteria, in this paper we are interested in objective functions of the forms
\begin{eqnarray}
 f_{\max} & = &  \max_{j=1,2,\ldots,n} f_j(C_j), \label{fmax} \\
 f_{\sum} & = & \sum_{j=1}^n f_j(C_j), \label{fsum}
\end{eqnarray}
with functions $f_j$ non-decreasing for $j=1,2,\ldots,n$.

In particular, the first part of the paper will be dedicated to the minimization of
the \emph{makespan} $C_{\max} = \max \left\{C_j \mid j= 1,2, \ldots, n \right\}$
and of the \emph{total completion time} $\sum C_j = \sum_{j=1}^n C_j$,
although the obtained algorithms will work for an arbitrary objective function
of type (\ref{fmax}) or (\ref{fsum}), if certain pre-conditions are met.

In the second part of the paper, we assume that for each job 
$J_j$ a \emph{weight} $w_j$ or a 
\emph{due date} $d_j$ is given.
We focus on the following traditional scheduling objectives:
\begin{itemize}
	\item the \emph{weighted total completion time} \\
	$\sum w_j C_j = \sum_{j=1}^n w_j C_j$,
	\item the \emph{maximum lateness} \\
	$L_{\max} = \max \left\{C_j - d_j \mid j = 1, 2, \ldots, n \right\}$, 
	\item the \emph{total tardiness} \\
	$\sum T_j = \sum_{j=1}^n T_j$,
	\item the \emph{number of late jobs} \\
	$\sum U_j = \sum_{j=1}^n U_j$,
	\item the \emph{weighted number of late jobs} \\
	$\sum w_j U_j = \sum_{j=1}^n w_j U_j$,
\end{itemize}
where
$$
T_j = \left\{
\begin{array}{ll}
C_j - d_j, & \text{if } C_j > d_j, \\
0, & \text{otherwise},
\end{array}
\right.$$
and 
$$
U_j = \left\{
\begin{array}{ll}
1, & \text{if } C_j > d_j, \\
0, & \text{otherwise.}
\end{array}
\right.
$$
Note that all these objective functions are \emph{regular}, that is, nondecreasing in each job completion time $C_j$.

Using the standard three-field notation for scheduling problems \citep{GrahamEtAl:ThreeFieldNot},
our problem is denoted as
$$Fm \mid r_j, p_{ij} = p_i, p\text{-batch}, b_i \mid f,$$
where $f$ is a function of the form defined above. We refer to the described scheduling model as \emph{proportionate flow shop of batching machines} and abbreviate it by \emph{PFB}.

Next, we provide an example in order to illustrate the problem setting.

\begin{example}\label{FirstExample}
	Consider a PFB instance with $m=2$ machines, $n=5$ jobs, maximum batch sizes $b_1=3$ and $b_2=4$, processing times $p_1=2$ and $p_2=3$, and release dates $r_1=r_2=0$, $r_3=r_4=1$, and $r_5=2$. Figure \ref{Fig:FirstExampleA} illustrates a feasible schedule for the instance as job-oriented Gantt chart.
	
	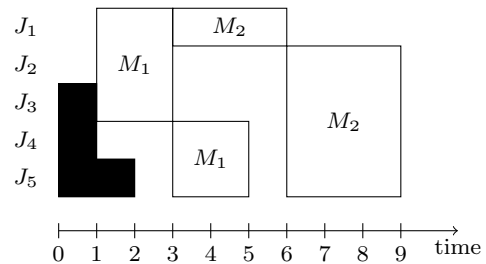
\begin{figure}[ht]
		\centering
		\begin{tikzpicture}
		\newcommand{\step}{0.5}
		\grid{9}{5}{\step}{north}
		\draw [fill=black] (0, \step) rectangle (\step, 3*\step);
		\draw [fill=black] (0, 0) rectangle (2*\step, \step);
		\draw [] (\step, 2*\step) rectangle (3*\step, 5*\step)  node [pos=0.5]{$M_1$};
		\draw [] (3*\step, 0) rectangle (5*\step, 2*\step) node [pos=0.5]{$M_1$};
		\draw [] (3*\step, 4*\step) rectangle (6*\step, 5*\step) node [pos=0.5]{$M_2$};
		\draw [] (6*\step, 0) rectangle (9*\step, 4*\step) node [pos=0.5]{$M_2$};
		\end{tikzpicture}
		\caption{A feasible example schedule.}
		\label{Fig:FirstExampleA}
	\end{figure}
	
	Each rectangle labeled by a machine represents a batch of jobs processed together on this machine. The black area indicates that the respective jobs have not been released at this time yet. Note that in this example none of the batches can be started earlier, since either a job of the batch has just arrived when the batch is started, or the machine is occupied before. Still, the schedule does not minimize the makespan, since the schedule shown in Figure \ref{Fig:FirstExampleB} is feasible as well and has a makespan of 8 instead of 9.
	
	\begin{figure}[ht]
		\centering
		\begin{tikzpicture}
		\newcommand{\step}{0.5}
		\grid{9}{5}{\step}{north}
		\draw [fill=black] (0, \step) rectangle (\step, 3*\step);
		\draw [fill=black] (0, 0) rectangle (2*\step, \step);
		\draw [] (0, 3*\step) rectangle (2*\step, 5*\step) node [pos=0.5]{$M_1$};
		\draw [] (2*\step, 0) rectangle (4*\step, 3*\step) node [pos=0.5]{$M_1$};
		\draw [] (2*\step, 4*\step) rectangle (5*\step, 5*\step) node [pos=0.5]{$M_2$};
		\draw [] (5*\step, 0) rectangle (8*\step, 4*\step) node [pos=0.5]{$M_2$};
		\end{tikzpicture}
		\caption{An example schedule minimizing the makespan.}
		\label{Fig:FirstExampleB}
	\end{figure}
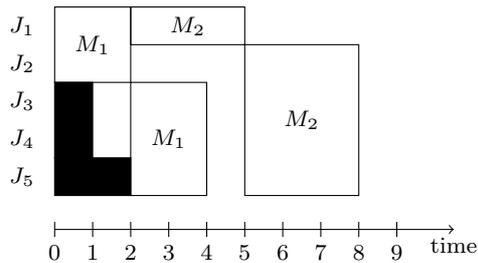
	
	The improvement in the makespan was achieved by reducing the size of the first batch on $M_1$ from three to two, which allows to start it one time step earlier. Observe that no job can finish $M_2$ before time step 5. Moreover, not all five jobs fit into one batch on $M_2$. Hence, at least two batches are necessary on $M_2$ in any schedule. Therefore, 8 is the optimal makespan and no further improvement is possible.
\end{example}

The rest of the paper is structured as follows. The remaining parts of this section provide an overview of the related literature and our results, respectively. In Section \ref{Sec:Perm} we prove that permutation schedules with jobs in earliest release date order are optimal for PFBs with the makespan and total completion time objectives. We also show that permutation schedules are not optimal for any other traditional scheduling objective functions. In Section \ref{Sec:Algorithm} we construct a dynamic program to find the optimal permutation schedule with a given, fixed job order. We show that, if the number of machines $m$ is fixed, the dynamic program can be used to minimize the makespan or total completion time in a PFB in polynomial time. In Section \ref{Sec:equalReleaseDates} we consider PFBs where all release dates are equal. We show that, in this special case, permutation schedules are always optimal and that the dynamic program from Section \ref{Sec:Algorithm} can be applied to solve most traditional scheduling objectives. Finally, in Section \ref{Sec:Conclusion} we draw conclusions, reflect on our results in the context of some additional existing literature, and point out future research directions.

\subsection{Literature}

In this overview as well as the remainder of this paper, we are exclusively concerned with flow shops of parallel batching machines, which we will simply call batching machines. For complexity results about (proportionate) flow shops with serial batching machines, we refer e.g.\ to \cite{Ng2007:SerialBatchingFlowShop} and \cite{Brucker2011:PolyTimeSerialBatching}. It is interesting to note that, although the different type of batching machines changes the problem significantly, the methodology used by \citet{Ng2007:SerialBatchingFlowShop} is quite similar to ours: for a fixed job permutation a dynamic program is developed that adds jobs one by one, resulting in polynomial runtime for a fixed number of machines.

The proportionate flow shop problem with batching machines as studied in this paper has so far not received a lot of attention from researchers, although several applications have appeared in the literature. 
In addition to the application from pharmacy described in Section \ref{Sec:Introduction}, \cite{SungEtAl:ProblemReduction} mention a variety of applications in the manufacturing
industry, e.g., in the production of semi-con\-duc\-tors.
Furthermore, there are several papers which study the scheduling sequences of locks for ships along a canal or river, which can in some ways be seen as a generalization of PFBs (see, e.g., \cite{Passchyn2019}). The problem is generally strongly NP-complete, but several special cases for which it can be solved in polynomial time have been identified.

Another stream of recent research focuses on finding practical scheduling techniques for production processes which can be viewed as NP-hard generalizations of the PFB problem. Usually application-specific heu\-ris\-tics are developed and their quality is certified by computational studies. A preliminary study of the specific industrial setup which motivated our research was done in \cite{ackermannGOR}. For studies from different areas of application we refer, for instance, to \citep{Li2015:HybridFlowShopBatchingGenetic, Zhou2016:SerialAndParallelBatchingEvolutionary, Tan2018:FF2IncompatibleWeightedTardiness, Chen2019:QuantumInspired, Li2019:SteelSerialAndParallelBatching}. The problems typically stem from various industries, e.g. electronics, steel, or equipment manufacturing. Features that make these problems difficult include mixtures of parallel and serial batching machines, families of jobs that are incompatible to be batched together, hybrid/flexible flow shop setups, as well as no-wait constraints. 

Despite this multitude of possible applications, for PFBs as introduced in Section \ref{Sec:Introduction} significant hardness results have, to the best of our knowledge, not been achieved at all. Polynomial time algorithms are only known for the special case with no more than two machines and, in most cases, without release dates \citep{Ahmadi:Batching,SungYoon:DynamicArrivalsTwoBPM,SungKim:TwoMachinesDueDates}.
 
The paper by \citet{SungEtAl:ProblemReduction} is, as far as we know, the only work in literature considering a PFB with arbitrarily many machines. They propose a reduction procedure to decrease the size of an instance by eliminating dominated machines. Using this procedure, they develop heuristics to minimize the makespan and the sum of completion times in a PFB and conduct a computational study certifying quality and efficiency of their heuristic approach. However, they do not establish any complexity result.

For the case with $m=2$ machines and no release dates, \citet{Ahmadi:Batching} show that the makespan can be minimized by completely filling all batches on the first machine (other than the last one, if the number of jobs is not a multiple of $b_1$) and completely filling all batches on the second machine (other than the first one, if the number of jobs is not a multiple of $b_2$).
This structural result immediately yields an $\mathcal{O}(n)$ time algorithm.
If release dates do not all vanish, i.e. $r_j >0$ for some $j=1,2, \ldots, n$, then the makespan can be minimized in $\mathcal{O}(n^2)$ time, see \citep{SungYoon:DynamicArrivalsTwoBPM}. The authors use a dynamic programming algorithm in order to find an optimal batching on the first machine. On the second machine, the strategy to completely fill all batches (other than maybe the first one) is still optimal.

For objectives other than the makespan, results are only known if all jobs are released at the same time.
In this case, the total completion time in a two-machine PFB can be minimized in $\mathcal{O}(n^3)$ time, see \citep{Ahmadi:Batching}.
\citet{SungKim:TwoMachinesDueDates} provide dynamic programs to minimize 
$L_{\max}$,
$\sum T_j$,
and
$\sum U_j$ in a two-machine PFB without release dates in polynomial time.

There are two special cases and one generalization of PFB which have been studied relatively extensively in the literature and which are of importance when it comes to comparison with the wider literature.
First, consider the special case of PFB with maximum batch size equal to one on all machines. This is the usual proportionate flow shop problem, which is well-researched and known to be polynomially solvable for many variants. An overview is given in \citep{Panwalkar:ProportionateReview}. In particular, the makespan and the total completion time can be minimized in $O(n \log n + nm)$ time, by ordering jobs according to nondecreasing release dates and scheduling each job as early as possible on each machine. Note that the term $n\log n$ is due to the sorting of the jobs while the term $nm$ arises from computing start and completion times for each operation of each job. 

Second, consider the special case of PFB where the number of machines is fixed to $m=1$, which leaves us with the problem of scheduling a single batching machine with identical (but not necessarily unit) processing times. This problem, too, has been investigated by various authors, see \citep{IkuraGimple:SingleBPM, LeeEtAl:SemiconductorBatching, Baptiste2000:BatchingIdenticalJobs, Condotta2010}. Even with release dates, the problem can be solved in polynomial time for almost all standard scheduling objective functions, except for the weighted total tardiness, for which it is open.
Note, in particular, the reduction procedure from \cite{SungEtAl:ProblemReduction} implies that for the cases where all batch sizes or all processing times are equal, $b_i = \bar{b}$ or $p_i = \bar{p}$ for all $i = 1, 2, \ldots, m$, the problem of scheduling a PFB reduces to scheduling a single batching machine with equal processing times.
Thus, for those two special cases, the results by \cite{Baptiste2000:BatchingIdenticalJobs} imply that scheduling a PFB is polynomially solvable for all objective functions other than the weighted total tardiness.

Third, in contrast to the two special cases from before, consider the generalization of PFB where processing times may be job- as well as machine-dependent. We end up with the usual flow shop problem with batching machines. Observe that this problem is also a generalization of traditional flow shop as well as scheduling a single batching machine.
Therefore, we can immediately deduce that, even without release dates, it is strongly NP-hard to minimize the makespan for any fixed number of machines $m \geq 3$ and to minimize the total completion time for any fixed number of machines $m \geq 2$ \citep{GareyEtAl:ComplexityFlowShop}. Furthermore, minimizing the maximum lateness (or, equivalently, the makespan with release dates) is strongly NP-hard even for the case with $m=1$ and $b_1 = 2$ \citep{BruckerEtAl:SchedulingABatchingMachine}.
Finally, \citet{PottsEtAl:ShopBatching} showed that scheduling a flow shop with batching machines to minimize the makespan, without release dates, is strongly NP-hard, even for $m=2$ machines. Recall that for traditional flow shop with $m=2$ machines a schedule with minimum makespan can be computed in $O(n \log n)$ time \citep{Johnson:TwoMachineFlowShop}.

In conclusion, we see that dropping either the batching or the multi-machine property of PFB leads to easy special cases, while dropping the proportionate property leads to a very hard generalization. Therefore, PFB lies exactly on the borderline between easy and hard problems, which makes it an interesting problem to study from a theoretical perspective, in addition to the practical considerations explained in the introduction.

\subsection{Our results}
Our main focus in this paper is the complexity of PFBs with an arbitrary, but fixed number $m$ of machines. This is motivated by the specific industrial application we consider: while the number of steps needed to produce individualized drugs in a process as described in the beginning is significantly larger than $2$ or $3$ (machine numbers which are often studied separately in the literature), it is not arbitrarily large and it does not change once the process is set up. Thus, it is reasonable to assume that $m$ is fixed and not part of the input.

\begin{table*} \centering%
	\begin{tabular}{lccc}
		\toprule
		Problem & $m=1$ & $m=2$ & fixed $m\geq 3$ \\
		\midrule
		$Fm \mid r_j, p_{ij} = p_i, p\text{-batch}, b_i \mid C_{\max}$ & \citet{IkuraGimple:SingleBPM} & \citet{SungYoon:DynamicArrivalsTwoBPM} & \textbf{Corollary \ref{Cor:FixedM}/\ref{Cor:Makespan}} \\
		$Fm \mid r_j, p_{ij} = p_i, p\text{-batch}, b_i \mid \sum C_j$ & \citet{Ahmadi:Batching} & \textbf{Corollary \ref{Cor:FixedM}} & \textbf{Corollary \ref{Cor:FixedM}} \\
		$Fm \mid p_{ij} = p_i, p\text{-batch}, b_i \mid \sum w_j C_j$ & \citet{Baptiste2000:BatchingIdenticalJobs} & \textbf{Theorem \ref{Thm:NoRelease}} & \textbf{Theorem \ref{Thm:NoRelease}} \\
		$Fm \mid p_{ij} = p_i, p\text{-batch}, b_i \mid L_{\max}$ & \citet{LeeEtAl:SemiconductorBatching} & \citet{SungKim:TwoMachinesDueDates} & \textbf{Theorem \ref{Thm:NoRelease}} \\
		$Fm \mid p_{ij} = p_i, p\text{-batch}, b_i \mid \sum T_j$ & \citet{Baptiste2000:BatchingIdenticalJobs} & \citet{SungKim:TwoMachinesDueDates} & \textbf{Theorem \ref{Thm:NoRelease}} \\
		$Fm \mid p_{ij} = p_i, p\text{-batch}, b_i \mid \sum U_j$ & \citet{LeeEtAl:SemiconductorBatching} & \citet{SungKim:TwoMachinesDueDates} & \textbf{Theorem \ref{Thm:wjUj}} \\
		$Fm \mid p_{ij} = p_i, p\text{-batch}, b_i \mid \sum w_j U_j$ & \citet{Baptiste2000:BatchingIdenticalJobs} & \textbf{Theorem \ref{Thm:wjUj}} & \textbf{Theorem \ref{Thm:wjUj}} \\
		\bottomrule
	\end{tabular}%
	\caption{Polynomially solvable PFB variants. To the best of our knowledge, each entry contains the first reference solving the corresponding problem or a generalization in polynomial time.}\label{Tab:PolySolvablePFBVariants}%
\end{table*}

This paper is, to the best of our knowledge, the first study which provides any complexity results -- positive or negative -- for PFBs with more than $m=2$ machines (see Table  \ref{Tab:PolySolvablePFBVariants} for an overview). As such we close several questions which have remained open for over fifteen years, despite the significant practical importance of PFBs in the area of production lines. 
We show that for PFBs with release dates and a fixed number $m$ of machines, the makespan and the total completion time can be minimized in polynomial time.
For PFBs without release dates, we further show that the weighted total completion time, the maximum lateness, the total tardiness and the (weighted) number of late jobs can be minimized in polynomial time.

Note that, while this paper closes many open complexity questions for PFBs, the proposed algorithm has practical limitations concerning its running time. We discuss this in some more detail also in the conclusions to this paper. Still, some of the structural results proved in this paper, in particular concerning the optimality of permutation schedules proved in Sections \ref{Sec:Perm} and \ref{Sec:equalReleaseDates}, can be of great help for the future design of heuristics or approximation algorithms, which run fast in practice. Indeed, for the specific industrial setup of the study this paper originates from, such practical heuristics have already been considered by \citet{ackermannGOR}.
However, a close investigation of heuristics or approximation algorithms for PFBs in general is beyond the scope of this paper.

\section{Optimality of Permutation Schedules}\label{Sec:Perm}

Our approach to scheduling a PFB relies on the well-known concept of permutation schedules. In a \emph{permutation schedule} the order of the jobs is the same on all machines of the flow shop. This means there exists a permutation $\sigma$ of the job indices such that $c_{i\sigma(1)}\leq c_{i\sigma(2)} \leq \ldots \leq c_{i\sigma(n)}$, for all $i=1,\ldots,m$. Due to job-independent processing times and since a machine can only process one batch at a time, clearly the same then also holds for starting times instead of completion times.
If there exists an optimal schedule $S^*$ which is a permutation schedule with a certain ordering $\sigma^*$ of the jobs, we say that \emph{permutation schedules are optimal}. A job ordering $\sigma^*$ which gives rise to an optimal permutation schedule is then called an \emph{optimal job ordering}.
Often, if the job ordering is clear, we will assume that jobs are already numbered with respect to that ordering and drop the notation of $\sigma$ or $\sigma^*$.
Finally, an ordering $\sigma$ of the jobs is called an \emph{earliest release date ordering}, if $r_{\sigma(1)} \leq r_{\sigma(2)} \leq \ldots \leq r_{\sigma(n)}$.

Suppose that for some problem 
$$Fm \mid r_j, p_{ij} = p_i, p\text{-batch}, b_i \mid f$$
permutation schedules are optimal.
Then the scheduling problem can be split into two parts: 
\begin{itemize}
	\item[(i)] find an optimal ordering $\sigma^*$ of the jobs;
	\item[(ii)] for each machine, partition the job set into batches and order those batches in accordance with the ordering $\sigma^*$, such that the resulting schedule is optimal.
\end{itemize}

In this section we deal with part (i) and show that for $f\in\{C_{\max},\sum C_j\}$ optimal job orderings exist and can be found in $O(n \log n)$ time.
Then, in Section \ref{Sec:Algorithm}, we show that under certain, very general preconditions part (ii) can be done in polynomial time via dynamic programming.

We begin with a technical lemma which will help us to show optimality of permutation schedules in this section and also in Section \ref{Sec:equalReleaseDates}.

\begin{lemma}
	\label{lem:releaseDatePermu}
	Let $S$ be a feasible schedule for a PFB and let $\sigma$ be some earliest release date ordering of the jobs.
	Then there exists a feasible permutation schedule $\hat{S}$ in which the jobs are ordered by $\sigma$ and the multi-set of job completion times in $\hat{S}$ is the same as in $S$.
\end{lemma}
\begin{proof}
	Suppose that jobs are indexed according to ordering $\sigma$, i.e.\ $\sigma$ is the identity permutation. Otherwise renumber jobs accordingly. Note that
	since $\sigma$ is an earliest release date ordering, it follows that 
	$r_1\leq r_2 \leq \ldots \leq r_n$.	
	Let $B^{(i)}_\ell$, $\ell=1,\ldots,\kappa^{(i)}$, be the $\ell$-th batch processed on $M_i$ in $S$, where $\kappa^{(i)}$ is the total number of batches used on $M_i$. Let $s^{(i)}_\ell$ and $c^{(i)}_\ell$ be the start and completion times of batch
	$B^{(i)}_\ell$.
	Furthermore, let $k_\ell^{(i)}=\abs{B^{(i)}_\ell}$ be the size of batch $B^{(i)}_\ell$ and let 
	$$j^{(i)}_\ell = \sum_{s = 1}^{\ell} k^{(i)}_s$$
	be the the number of jobs
	processed in batches $B^{(i)}_1$, $B^{(i)}_2, \ldots, B^{(i)}_\ell$.
	
	Construct the new schedule $\hat{S}$ using batches $\hat{B}^{(i)}_\ell$, $i = 1,2, \ldots, m$, $\ell = 1,2,
	\ldots, \kappa^{(i)}$, with start times $\hat{s}^{(i)}_\ell$ and completion times $\hat{c}^{(i)}_\ell$ given by
	\begin{eqnarray*}
		\hat{B}^{(i)}_\ell & = & \{ J_{j^{(i)}_{\ell-1} + 1}, J_{j^{(i)}_{\ell-1} + 2}, \ldots, J_{j^{(i)}_{\ell}}\},\\
		\hat{s}^{(i)}_\ell & = & s^{(i)}_\ell, \\
		\hat{c}^{(i)}_\ell & = & c^{(i)}_\ell.
	\end{eqnarray*}
	
	In other words, in schedule $\hat{S}$ we use batches with the same start and completion times and of the same size as in $S$. 
	We only reassign which job is processed in which batch in such a way, that
	the first batch on machine $M_i$ processes the first $k^{(i)}_1$ jobs (in order of their numbering, i.e.\ in the earliest release date order),
	the second batch processes jobs $J_{k^{(i)}_1+1}$ to $J_{k^{(i)}_1+k^{(i)}_2}$,
	the third batch processes jobs  $J_{k^{(i)}_1 + k^{(i)}_2+1}$ to $J_{k^{(i)}_1 + k^{(i)}_2 + k^{(i)}_3}$  and so on.
	Such, in schedule $\hat{S}$, jobs are processed in order of their indices on all machines (i.e.\ in order $\sigma$).
	
	Clearly, as all batches have exactly the same start and completion times, as well as the same sizes as in $S$, for schedule $\hat{S}$ the multi-set of job completion times on the last machine is exactly the same as for $S$. 
	
	It remains to show that $\hat{S}$ is feasible, i.e., no job is started on machine $M_1$ before its
	release date and no job is started on a machine $M_i$, $i = 2, 3, \ldots, m$ before finishing processing on machine $M_{i-1}$.
	
	We first show that no batch $\hat{B}^{(1)}_\ell$ on machine $M_1$ violates any release dates.
	Indeed, due to the feasibility of $S$, at the time $s^{(1)}_\ell$, when batch $B^{(1)}_\ell$ is started, at least $j_\ell^{(1)}$ jobs are released (recall that $j_\ell^{(1)}$ is the number of jobs
	processed in batches $B^{(1)}_1, B^{(1)}_2, \ldots, B^{(1)}_\ell$).
	This means that in particular the first $j_\ell^{(1)}$ jobs in earliest release date order are released at time
	$s^{(1)}_\ell$. 
	Thus, since batch $\hat{B}^{(1)}_\ell$ by construction contains only the jobs
	$J_{j_{\ell-1}^{(1)} +1}$ to $J_{j_\ell^{(1)}}$ and since jobs are numbered in earliest release date order, batch $\hat{B}^{(1)}_\ell$ does not violate any release dates.
	
	Finally, using an analogous argument, we show that no batch starts processing on machine $M_i$ before all jobs it contains finish processing on machine $M_{i-1}$, $i = 2, 3, \ldots, m$.
	By time $s^{(i)}_\ell$, when batch $B^{(i)}_\ell$ is started on machine $M_i$ in schedule $S$, machine $M_{i-1}$ has processed at least
	$j_\ell^{(i)}$ jobs.
	Thus, by construction, the same is true for batch $\hat{B}^{(i)}_\ell$ in schedule $\hat{S}$.
	In particular, this means that at time $s^{(i)}_\ell$ the first $j_\ell^{(i)}$ jobs in order of their numbering are finished on machine $M_{i-1}$ in schedule $\hat{S}$.
	Since batch $\hat{B}^{(i)}_\ell$ contains only the jobs
	$J_{j_{\ell-1}^{(i)} +1}$ to $J_{j_\ell^{(i)}}$,
	this means that batch $\hat{B}^{(i)}_\ell$ starts only after each job contained in it has finished processing on machine $M_{i-1}$.
	Therefore, schedule $\hat{S}$ is feasible.\qed
\end{proof}
Note that, in particular, if all release dates are equal, then any ordering $\sigma$ is an earliest release date ordering of the jobs. This fact will be used in Section \ref{Sec:equalReleaseDates}.

Now we show that for a PFB with the makespan or total sum of completion times objective, permutation schedules are optimal and that any earliest release date ordering is an optimal ordering of the jobs. 
This result generalizes \citep[Proposition 1]{SungYoon:DynamicArrivalsTwoBPM} to arbitrarily many machines. For makespan and total completion time, it also extends \citep[Lemma 1]{SungEtAl:ProblemReduction} to the case with release dates.

\begin{theorem}\label{Thm:PermI}
	For a PFB with objective function $C_{\max}$ or $\sum C_j$, permutation schedules are optimal. Moreover, any earliest release date ordering is an optimal ordering of the jobs.
\end{theorem}
\begin{proof}
	Let $S^*$ be an optimal schedule with respect to makespan or total completion time. Let $\sigma$ be any earliest release date ordering. Using Lemma $\ref{lem:releaseDatePermu}$, construct a new permutation schedule $\hat{S}$ with jobs ordered by $\sigma$ on all machines and with the same multi-set of completion times as $S$.
	Since makespan and total completion time only depend on the multi-set of completion times, $\hat{S}$ is optimal, too.\qed
\end{proof}

Hence, when minimizing the makespan or the total completion time in a PFB, it is valid to limit oneself to permutation schedules in an earliest release date order.

To conclude this section, we present an example where no permutation schedule is optimal for a bunch of other objective functions. This shows that Theorem~\ref{Thm:PermI} does not hold for these objective functions. This also implies that the statements made in \citep[Proposition 1]{SungYoon:DynamicArrivalsTwoBPM}, \citep[Lemma 1]{SungEtAl:ProblemReduction}, and \citep[Lemma 1]{SungKim:TwoMachinesDueDates} cannot be generalized to settings in which release dates and due dates or release dates and weights are present.

\begin{example}\label{Exa:NonPermSchedule}
Consider a PFB with $m=3$ machines, $b_1=p_1=b_3=p_3=1$, and $b_2=p_2=2$. Suppose there are $n=2$ jobs with release dates $r_1=0$ and $r_2=1$. Let due dates $d_1=6$ and $d_2=5$ be given. First, we show that no permutation schedule is optimal for $L_{\max}$. Figure \ref{Fig:NonPermSchedule} shows a job-oriented Gantt chart of a feasible schedule which is not a permutation schedule. Each rectangle labeled by a machine denotes a batch processed on this machine. The black box indicates that the second job has not been released yet. The gray shaded area denotes points in time after the due dates.

\begin{figure}[ht]
	\centering
	\begin{tikzpicture}
	\newcommand{\step}{0.5}
	\grid{7}{2}{\step}{north}
	\draw [fill=black](0*\step, 0*\step) rectangle (1*\step, 1*\step);
	\draw [color=\dueColor, fill=\dueColor] (6*\step, 1*\step) rectangle (8*\step, 2*\step);
	\draw [color=\dueColor, fill=\dueColor] (5*\step, 0*\step) rectangle (8*\step, 1*\step);
	\draw (0*\step, 1*\step) rectangle (1*\step, 2*\step) node [pos=0.5]{$M_1$};
	\draw (1*\step, 0*\step) rectangle (2*\step, 1*\step) node [pos=0.5]{$M_1$};
	\draw (2*\step, 0*\step) rectangle (4*\step, 2*\step) node [pos=0.5]{$M_2$};
	\draw (4*\step, 0*\step) rectangle (5*\step, 1*\step) node [pos=0.5]{$M_3$};
	\draw (5*\step, 1*\step) rectangle (6*\step, 2*\step) node [pos=0.5]{$M_3$};
	\end{tikzpicture}
	\caption{A non-permutation schedule for the instance of Example \ref{Exa:NonPermSchedule}.}
	\label{Fig:NonPermSchedule}
\end{figure}
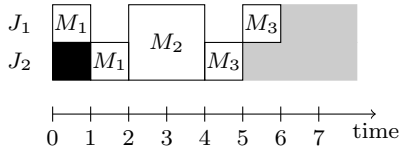

	When constructing a permutation schedule for this instance, two decisions have to be made, namely which of the two jobs is processed first, and whether a batch of size two or two batches of size one should be used on $M_2$. The resulting four possible schedules are depicted in Figure \ref{Fig:SuboptimalPermSchedules}.

\begin{figure}[ht]
	\centering
		\begin{tikzpicture}
		\newcommand{\step}{0.5}
		\grid{7}{2}{\step}{north}
		\draw [fill=black] (0*\step, 0*\step) rectangle (1*\step, 1*\step);
		\draw [color=\dueColor, fill=\dueColor] (6*\step, 1*\step) rectangle (8*\step, 2*\step);
		\draw [color=\dueColor, fill=\dueColor] (5*\step, 0*\step) rectangle (8*\step, 1*\step);
		\draw (0*\step, 1*\step) rectangle (1*\step, 2*\step) node [pos=0.5]{$M_1$};
		\draw (1*\step, 0*\step) rectangle (2*\step, 1*\step) node [pos=0.5]{$M_1$};
		\draw (2*\step, 0*\step) rectangle (4*\step, 2*\step) node [pos=0.5]{$M_2$};
		\draw (4*\step, 1*\step) rectangle (5*\step, 2*\step) node [pos=0.5]{$M_3$};
		\draw (5*\step, 0*\step) rectangle (6*\step, 1*\step) node [pos=0.5]{$M_3$};
		\end{tikzpicture}
		\vspace{1em}\\
		\begin{tikzpicture}
		\newcommand{\step}{0.5}
		\grid{7}{2}{\step}{north}
		\draw [fill=black] (0*\step, 0*\step) rectangle (1*\step, 1*\step);
		\draw [color=\dueColor, fill=\dueColor] (6*\step, 1*\step) rectangle (8*\step, 2*\step);
		\draw [color=\dueColor, fill=\dueColor] (5*\step, 0*\step) rectangle (8*\step, 1*\step);
		\draw (0*\step, 1*\step) rectangle (1*\step, 2*\step) node [pos=0.5]{$M_1$};
		\draw (1*\step, 0*\step) rectangle (2*\step, 1*\step) node [pos=0.5]{$M_1$};
		\draw (1*\step, 1*\step) rectangle (3*\step, 2*\step) node [pos=0.5]{$M_2$};
		\draw (3*\step, 0*\step) rectangle (5*\step, 1*\step) node [pos=0.5]{$M_2$};
		\draw (3*\step, 1*\step) rectangle (4*\step, 2*\step) node [pos=0.5]{$M_3$};
		\draw (5*\step, 0*\step) rectangle (6*\step, 1*\step) node [pos=0.5]{$M_3$};
		\end{tikzpicture}
		\vspace{1em}\\
		\begin{tikzpicture}
		\newcommand{\step}{0.5}
		\grid{7}{2}{\step}{north}
		\draw [fill=black] (0*\step, 0*\step) rectangle (1*\step, 1*\step);
		\draw [color=\dueColor, fill=\dueColor] (6*\step, 1*\step) rectangle (8*\step, 2*\step);
		\draw [color=\dueColor, fill=\dueColor] (5*\step, 0*\step) rectangle (8*\step, 1*\step);
		\draw (2*\step, 1*\step) rectangle (3*\step, 2*\step) node [pos=0.5]{$M_1$};
		\draw (1*\step, 0*\step) rectangle (2*\step, 1*\step) node [pos=0.5]{$M_1$};
		\draw (3*\step, 0*\step) rectangle (5*\step, 2*\step) node [pos=0.5]{$M_2$};
		\draw (6*\step, 1*\step) rectangle (7*\step, 2*\step) node [pos=0.5]{$M_3$};
		\draw (5*\step, 0*\step) rectangle (6*\step, 1*\step) node [pos=0.5]{$M_3$};
		\end{tikzpicture}
		\vspace{1em}\\
		\begin{tikzpicture}
		\newcommand{\step}{0.5}
		\grid{7}{2}{\step}{north}
		\draw [fill=black] (0*\step, 0*\step) rectangle (1*\step, 1*\step);
		\draw [color=\dueColor, fill=\dueColor] (6*\step, 1*\step) rectangle (8*\step, 2*\step);
		\draw [color=\dueColor, fill=\dueColor] (5*\step, 0*\step) rectangle (8*\step, 1*\step);
		\draw (2*\step, 1*\step) rectangle (3*\step, 2*\step) node [pos=0.5]{$M_1$};
		\draw (1*\step, 0*\step) rectangle (2*\step, 1*\step) node [pos=0.5]{$M_1$};
		\draw (4*\step, 1*\step) rectangle (6*\step, 2*\step) node [pos=0.5]{$M_2$};
		\draw (2*\step, 0*\step) rectangle (4*\step, 1*\step) node [pos=0.5]{$M_2$};
		\draw (6*\step, 1*\step) rectangle (7*\step, 2*\step) node [pos=0.5]{$M_3$};
		\draw (4*\step, 0*\step) rectangle (5*\step, 1*\step) node [pos=0.5]{$M_3$};
		\end{tikzpicture}
	\caption{All possible permutation schedules without unnecessary idle time for the instance of Example \ref{Exa:NonPermSchedule}.}
	\label{Fig:SuboptimalPermSchedules}
\end{figure}
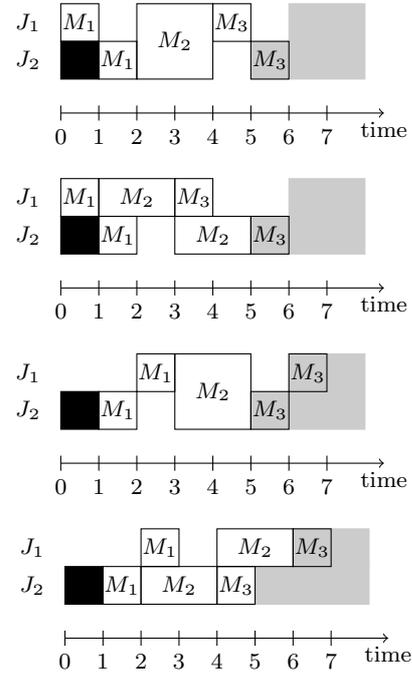

However, the schedule in Figure \ref{Fig:NonPermSchedule} achieves an objective value of zero, while all schedules in Figure \ref{Fig:SuboptimalPermSchedules} have a positive objective value. Hence, there is no optimal permutation schedule for this instance. The same example shows that there is no optimal permutation schedule for $\sum T_j$ and $\sum U_j$.

Choosing weights $w_1=1$, $w_2=3$ leads to an instance where no permutation schedule is optimal for $\sum w_jC_j$, since the schedule in Figure \ref{Fig:NonPermSchedule} achieves an objective value of $6\cdot1+5\cdot3=21$, while all schedules in Figure \ref{Fig:SuboptimalPermSchedules} have an objective value of at least $22$.

\end{example}

\section{Dynamic Programming to Find an Optimal Schedule for a Given Job Permutation \label{Sec:Algorithm}} 

In this section we show that for a fixed number $m$ of machines and a fixed ordering of the jobs $\sigma$, we can construct a polynomial time dynamic program, which finds a permutation schedule that is optimal among all schedules obeying job ordering $\sigma$. The algorithm works for an arbitrary regular sum or bottleneck objective function. Combined with Theorem \ref{Thm:PermI} this shows that the makespan and the total completion time in a PFB can be minimized in polynomial time for fixed $m$.

The dynamic program is based on the important observation that, for a given machine $M_i$, the set of possible job completion times on $M_i$ is not too large. In order to formalize this statement, we introduce the notion of a schedule being $\Gamma$-active. For ease of notation, we use the shorthand $[k]=\{1,2,\ldots,k\}$.
\begin{definition}\label{Def:GammaActive}
	Given an instance of PFB with $n$ jobs and $m$ machines, let
	\begin{align*}
	\Gamma_i=\left\{r_{j'} + \sum_{i'=1}^{i}\lambda_{i'}p_{i'}\st j'\in[n],\ \lambda_{i'}\in[n]\text{ for } i'\in[i] \right\}
	\end{align*}
	for all $i \in \left[m\right]$.
	We say that a schedule $S$ is \emph{$\Gamma$-active}, if $c_{ij}(S) \in \Gamma_i$, for any job $J_j$, $j \in \left[n\right]$, and any machine $M_i$, $i \in \left[m\right]$.		
\end{definition}
Note that  on machine $M_i$, there are only $\abs{\Gamma_i}\leq n^{i+1}\leq n^{m+1}$ possible job completion times to consider for any $\Gamma$-active schedule $S$.

Now we show that for a PFB problem with a regular objective function, any schedule can be transformed into a $\Gamma$-active schedule, without increasing the objective value. To do so, we prove that this is true even for a slightly stronger concept than $\Gamma$-active.
\begin{definition}\label{Def:BatchActive}
	A schedule is called \emph{batch-active}, if no batch can be started earlier without violating feasibility of the schedule and without changing the order of batches.
\end{definition}

Clearly, given a schedule $S$ that is not batch-active, by successively removing unnecessary idle times we can obtain a new, batch-active schedule $S^\prime$. Furthermore, for regular objective functions, $S^\prime$ has objective value no higher than the original schedule $S$. These two observations immediately yield the following lemma.

\begin{lemma}\label{Lem:BatchActive}
	A schedule for a PFB can always be transformed into a batch-active schedule such that any regular objective function is not increased. This transformation does not change the order in which the jobs are processed on the machines.\qed
\end{lemma}

Now we show that, indeed, being batch-active is stronger than being $\Gamma$-active, in other words that every batch-active schedule is also $\Gamma$-active. This result generalizes an observation made by \citet[Section~2.1]{Baptiste2000:BatchingIdenticalJobs} from a single machine to the flow shop setting.

\begin{lemma}\label{Lem:ActiveComplTimes}	
	In a PFB, any batch-active schedule is also $\Gamma$-active.
\end{lemma}
\begin{proof}
	Fix some $i\in[m]$ and $j\in[n]$. Let $B^{(i)}_\ell$ be the batch which contains job $J_j$ on machine $M_i$. Since the schedule is batch-active, $B^{(i)}_\ell$ is either started at the completion time of the previous batch $B^{(i)}_{\ell-1}$ on the same machine or as soon as all jobs of $B^{(i)}_\ell$ are available on machine $M_i$. In the former case, all jobs $J_{j'}\in B^{(i)}_{\ell-1}$ satisfy $c_{ij} = c_{ij'}+p_i$. In the latter case, there is a job $J_{j''}\in B^{(i)}_\ell$ such that $c_{ij}=c_{(i-1)j''}+p_i$, where we write $c_{0j''}=r_{j''}$ for convenience. The claim follows inductively by observing that the former case can happen at most $n-1$ times in a row, since there are at most $n$ batches on each machine.\qed
\end{proof}

Together, Lemmas \ref{Lem:BatchActive} and \ref{Lem:ActiveComplTimes} imply the following desired property.
\begin{lemma} \label{Lem:optSchedGammaI}
	Let $S$ be a schedule for a PFB problem with regular objective function $f$. Then there exists a schedule $S^\prime$ for the same PFB problem, such that
	\begin{enumerate}
		\item on each machine, jobs appear in the same order in $S^\prime$ as they do in $S$,
		\item $S^\prime$ has objective value no larger than $S$, and
		\item $S'$ is $\Gamma$-active.\qed
	\end{enumerate}	
\end{lemma}
In particular, Lemma \ref{Lem:optSchedGammaI} shows that, if for a PFB problem an optimal job ordering $\sigma^*$ is given, then there exists an optimal schedule which is a permutation schedule with jobs ordered by $\sigma^*$ and in addition $\Gamma$-active.

From this point on, we assume that the objective function is a regular sum or bottleneck function, that is, $f(C) = \bigoplus_{j=1}^n f_j(C_j)$, where $\bigoplus\in\{\sum,\max\}$ and $f_j$ is nondecreasing for all $j\in[n]$. We also use the symbol $\oplus$ as a binary operator.
In what follows, we present a dynamic program which, given a job ordering $\sigma$, finds a schedule that is optimal among all $\Gamma$-active permutation schedules obeying job ordering $\sigma$.
For simplicity, we assume that jobs are already indexed by $\sigma$. The dynamic program schedules the jobs one after the other, until all jobs are scheduled. Due to Lemma \ref{Lem:optSchedGammaI}, if $\sigma$ is an optimal ordering, then the resulting schedule is optimal.

Given an instance $I$ of a PFB and a job index $j\in[n]$, let $I_j$ be the modified instance which contains only the jobs $J_1,\dots,J_j$. We write $\Gamma=\Gamma_1 \times \Gamma_2 \times \ldots \times \Gamma_m$ and $\mathcal{B}=\left[b_1\right] \times \left[b_2\right] \times \ldots \times \left[b_m\right]$, where $\times$ is the standard Cartesian product. For a vector
\[t = (t_1, t_2, \ldots, t_i, \ldots, t_{m-1}, t_m) \in \Gamma\]
of $m$ possible completion times and a vector
\[ k = (k_1, k_2, \ldots, k_i, \ldots, k_{m-1}, k_m) \in \mathcal{B}\]
of $m$ possible batch sizes, we say that a schedule $S$ for instance $I_j$ \emph{corresponds} to $t$ and $k$ if, for all $i\in[m]$, $c_{ij}=t_i$ and job $J_j$ is contained in a batch with exactly $k_i$ jobs (including $J_j$) on $M_i$.

Next, we define the variables $g$ of the dynamic program. Let $\mathcal{S}(j,t,k)$ be the set of feasible permutation schedules $S$ for instance $I_j$ satisfying the following properties:
\begin{enumerate}
	\item jobs in $S$ are ordered by their indices,
	\item $S$ corresponds to $t$ and $k$, and
	\item $S$ is $\Gamma$-active.
\end{enumerate}
Then, for $j\in[n]$, $t\in\Gamma$, and $k\in\mathcal{B}$, we define
\[g(j,t,k)=g(j,t_1,t_2,\dots,t_m,k_1,k_2,\dots,k_m)\]
to be the minimum objective value of a schedule in $\mathcal{S}(j,t,k)$ for instance $I_j$. If no such schedule exists, the value of $g$ is defined to be $+\infty$.

Using the above definitions as well as Lemma \ref{Lem:optSchedGammaI}, the objective value of the best permutation schedule ordered by the job indices is given by $\min_{t,k} g(n,t,k)$, $t\in\Gamma$, $k\in\mathcal{B}$. 
In Lemmas \ref{Lem:DynStart} and \ref{Lem:DynRecurrence} (below) we provide formulas to recursively compute the values $g(j,t,k)$, $j=1,2,\ldots,n$, $t\in \Gamma$, $k \in \mathcal{B}$.
In total, we then obtain Algorithm \ref{Alg:DynProg} as our dynamic program.

\begin{algorithm}[!ht]
	\caption{
	}
	\label{Alg:DynProg}
	\begin{algorithmic}[1]
		\item[\begin{tabular}{ll}
			\textbf{Input:} & A PFB instance with regular sum or \\
			& bottleneck objective function $f=\bigoplus_{j=1}^n f_j$.\\
			\textbf{Output:} & An optimal schedule among all permutation\\
			& schedules ordered by the job indices.
		\end{tabular}]
		\item[]
		
		\State For each $t\in\Gamma$ and $k\in\mathcal{B}$, compute $g(1,t,k)$ via the formula defined in Lemma \ref{Lem:DynStart}.\label{Line:Start}
		\State For each $j=2,3,\dots,n$, $t\in\Gamma$, and $k\in\mathcal{B}$, compute $g(j,t,k)$ via the formula defined in Lemma \ref{Lem:DynRecurrence}. If $g(j,t,k)<+\infty$, also store a reference to the minimizers $t'$ and $k'$ on the right-hand side of the recurrence in Lemma \ref{Lem:DynRecurrence}.\label{Line:Recurrence}
		\State Find $t$ and $k$ minimizing $g(n,t,k)$ and start from them to backtrack the references stored in Step \ref{Line:Recurrence} in order to find the optimal schedule.\label{Line:Backtrack}
	\end{algorithmic}
\end{algorithm}

The remainder of this section deals with proving the correctness and running time bound of Algorithm \ref{Alg:DynProg}. 
As mentioned above, in Lemmas \ref{Lem:DynStart} and \ref{Lem:DynRecurrence} we define and prove the correctness of a recursive formula to compute all values of function $g$. Lemma \ref{Lem:DynStart} deals with the starting values $g(1,t,k)$ while Lemma \ref{Lem:DynRecurrence} deals with the recurrence relation.
Finally, in Theorem \ref{Thm:ConstantM} we connect the results of this section to prove the correctness of Algorithm \ref{Alg:DynProg}.

\begin{lemma}\label{Lem:DynStart}
	For $j=1$, $t_i\in\Gamma_i$, and $k_i\in[b_i]$, $i\in[m]$, the starting values of $g$ are given by
	\begin{align*}
	g(1,t,k) = \left\{
	\begin{array}{ll}
	f_1(t_m), & \text{if conditions (i)--(iii) hold},\\
	+\infty,  & \text{otherwise,}
	\end{array} \right.
	\end{align*}
	where 
	\begin{enumerate}[leftmargin=0.75cm, label=(\roman{enumi})]
		\item $k_i=1$ for all $i\in[m]$,
		\item $t_1\geq r_1 + p_1$, and
		\item $t_{i+1}\geq t_{i}+p_{i+1}$ for all $i\in[m-1]$.
	\end{enumerate}
\end{lemma}
\begin{proof}
	Conditions (i)--(iii) are necessary for the existence of a schedule in $\mathcal{S}(1,t,k)$ because $I_1$ consists of only one job and there must be enough time to process this job on each machine. Conversely, if (i)--(iii) are satisfied, then the vector $t$ of completion times uniquely defines a feasible schedule $S \in \mathcal{S}(1,t,k)$ by $c_{i1}(S)=t_i$. Since $S$ is uniquely defined by $t$ and thus the only schedule in $\mathcal{S}(1,t,k)$, its objective value $f_1(C_1)=f_1(c_{m1})=f_1(t_m)$ is optimal.\qed
\end{proof}

Now we turn to the recurrence formula to calculate $g$ for $j>1$ from the values of $g$ for $j-1$.

\begin{lemma}\label{Lem:DynRecurrence}
	For $j>1$, $t_i\in\Gamma_i$, and $k_i\in[b_i]$, $i\in[m]$, the values of $g$ are determined by
	\begin{align*}
	&g(j,t,k) \\
	&\quad = \left\{
	\begin{array}{ll}
	\min\{f_j(t_m)\oplus g(j-1,t',k') \mid (**) \}, & \text{if }(*),\\
	+\infty,  & \text{otherwise,}
	\end{array} \right.
	\end{align*}
	where the minimum over the empty set is defined to be $+\infty$. Here, $(*)$ is given by conditions
	\begin{enumerate}[leftmargin=0.75cm, label=(\roman{enumi})]
		\item $t_1\geq r_j + p_1$ and
		\item $t_{i+1}\geq t_{i}+p_{i+1}$ for all $i\in[m-1]$,
	\end{enumerate}
	and $(**)$ is given by conditions 
	\begin{enumerate}[leftmargin=0.75cm, label=(\roman{enumi}), resume]
		\item $t'_i\in\Gamma_i$,
		\item $k'_i\in[b_i]$,
		\item if $k_i=1$, then $t'_i\leq t_i-p_i$, and
		\item if $k_i>1$, then $t'_i=t_i$ and $k'_i=k_i-1$,
	\end{enumerate}
	for all $i\in[m]$.
\end{lemma}
\begin{proof}
Fix values $t_i\in\Gamma_i$ and $k_i\in[b_i]$, $i\in[m]$.
The conditions of $(*)$ are necessary for the existence of a schedule in $\mathcal{S}(j,t,k)$ because there must be enough time to process job $J_j$ on each machine. Therefore, $g$ takes the value $+\infty$, if $(*)$ is violated. For the remainder of the proof, assume that $(*)$ is satisfied. Hence, we have to show that 
\begin{equation}\label{eqn:AlgoStep}
g(j,t,k)=\min\left\{f_j(t_m)\oplus g(j-1,t',k') \st (**)\right\}.
\end{equation}

We first prove ``$\geq$''. If the left-hand side of (\ref{eqn:AlgoStep}) equals infinity, then this direction follows immediately. Otherwise, by the definition of $g$, there must be a schedule $S\in\mathcal{S}(j,t,k)$ with objective value $g(j,t,k)$.
Schedule $S$ naturally defines a feasible schedule $\myModifier{S}$ for instance $I_{j-1}$ by ignoring job $J_j$. 	

Observe that, because $S$ belongs to $\mathcal{S}(j,t,k)$ and therefore is $\Gamma$-active, $\myModifier{S}$ is also $\Gamma$-active and job $J_{j-1}$ finishes processing on machine $M_i$ at some time $\myModifier{t}_i \in \Gamma_i$. Also, since $S$ is feasible, on each machine $M_i$ job $J_{j-1}$ is scheduled in some batch of size $\myModifier{k}_i \leq b_i$.
Thus, $\myModifier{S}$ corresponds to two unique vectors $\myModifier{t} = (\myModifier{t}_1, \myModifier{t}_2, \ldots,
\myModifier{t}_m)$ and $\myModifier{k} = (\myModifier{k}_1, \myModifier{k}_2, \ldots, \myModifier{k}_m)$, which satisfy (iii) and (iv) from $(\ast \ast)$. Note that, in particular, $\myModifier{S} \in \mathcal{S}(j-1,\myModifier{t},\myModifier{k})$.

Furthermore, $\myModifier{t}$ and $\myModifier{k}$ also satisfy (v) and (vi) from $(\ast \ast)$. Indeed, due to the fixed job permutation, one of the following two things happens on each machine $M_i$ in $S$: either jobs $J_j$ and $J_{j-1}$ are batched together, or job $J_j$ is batched in a singleton batch. In the former case, it follows that $1<k_i=\myModifier{k}_i+1$ and $\myModifier{t}_i=t_i$, while the latter case requires $k_i=1$ and $t_i\geq \myModifier{t}_i + p_i$, since the machine is occupied by the previous batch with job $J_{j-1}$ until then.

Thus, $\myModifier{t}$ and $\myModifier{k}$ satisfy $(\ast \ast)$ and we obtain
\begin{align*}
g(j,t,k) &= \bigoplus_{j'=1}^j f_{j'}(C_{j'}(S)) \\
&= f_j(t_m)\oplus\bigoplus_{j'=1}^{j-1} f_{j'}(C_{j'}(\myModifier{S}))\\
&\geq f_j(t_m)\oplus g(j-1,\myModifier{t},\myModifier{k}),
\end{align*}
where the last inequality follows due to the definition of $g$ and $\myModifier{S}\in\mathcal{S}(j-1,\myModifier{t},\myModifier{k})$. Hence, the ``$\geq$'' direction in (\ref{eqn:AlgoStep}) follows because $\myModifier{t}$ and $\myModifier{k}$ satisfy $(**)$.

For the ``$\leq$'' direction, if the right-hand side of (\ref{eqn:AlgoStep}) equals infinity, then this direction follows immediately. Otherwise, let $\myModifier{t}$ and $\myModifier{k}$ be minimizers at the right-hand side. By definition of $g$ there must be a schedule $\myModifier{S}\in\mathcal{S}(j-1,\myModifier{t},\myModifier{k})$ for $I_{j-1}$ with objective value $g(j-1,\myModifier{t},\myModifier{k})$. 

We now show that $\myModifier{S}$ can be extended to a feasible schedule $S\in\mathcal{S}(j,t,k)$. Construct $S$ from $\myModifier{S}$ by adding job $J_j$ in the following way:
\begin{itemize}
	\item if in $\myModifier{S}$ there is a batch on machine $M_i$ which ends at time $t_i$, add $J_j$ to that batch (we show later that this does not cause the batch to exceed the maximum batch size $b_i$ of $M_i$);
	\item otherwise, add a new batch on machine $M_i$ finishing at time $t_i$ and containing only job $J_j$ (we show later that this does not create overlap with any other batch on machine $M_i$).
\end{itemize}	

First note that $c_{ij}(S) = t_i$ for all $i \in \left[m\right]$, and thus $S$ corresponds to $t$ by definition.	
Furthermore, for each machine $M_i$ consider the two cases (a) $k_i = 1$ and (b) $k_i > 1$.

In case (a), due to (v) it follows that $\myModifier{t}_i \leq t_i - p_i$. Since job $J_{j-1}$ finishing at time $\myModifier{t}_i$ is the last job completed on machine $M_i$ in schedule $\myModifier{S}$, by construction in schedule $S$ job $J_j$ is in a singleton batch on machine $M_i$ that starts at time $t_i - p_i$ and ends at time $t_i$.
Therefore, in this case, $S$ corresponds to $k$, because $k_i = 1$ and job $J_j$ is in a singleton batch. Also, since the last batch on machine $M_i$ in $\myModifier{S}$ ends at time $\myModifier{t}_i \leq t_i - p_i$ and the batch with job $J_j$ starts at time $t_i-p_i$, no overlapping happens between the new batch for job $J_j$ and any other batch on machine $M_i$.

In case (b), due to (vi) it follows that $\myModifier{t}_i = t_i$ and by construction job $J_j$ is scheduled in the same batch as job $J_{j-1}$ on machine $M_i$ in $S$. Since $\myModifier{S}$ corresponds to $\myModifier{k}$ and again due to (vi), this means that $J_j$ is scheduled in a batch with $k_i$ jobs on machine $M_i$ and S corresponds to $k$ in this case. Also, since $k_i \in \left[b_i\right]$ by definition, no batch in $S$ exceeds its permissible size.

Combining the considerations for cases (a) and (b), together with the feasiblity of schedule $\myModifier{S}$, it follows that
\begin{description}
	\item[(1)] $S$ corresponds to $k$,
	\item[(2)] no overlapping happens between batches in $S$, and
	\item[(3)] all batches in $S$ are of permissible size.
\end{description}	

In order to show feasibility of $S$ it is only left to show that no job starts before its release date on machine $M_1$, and no job starts on machine $M_i$ before it is completed on machine $M_{i-1}$.
As $\myModifier{S}$ is feasible, this is clear for all jobs other than $J_j$. On the other hand, since $S$ corresponds to $t$, and $t$ fulfills $(\ast)$, it also follows for job $J_j$.

Thus, $S$ is indeed feasible. In total, we have shown all conditions for $S\in\mathcal{S}(j,t,k)$. Therefore, we obtain
\begin{align*} \label{eq:recurranceLeq}
g(j,t,k) &\leq \bigoplus_{j'=1}^j f_{j'}(C_{j'}(S))\\
&= f_j(t_m)\oplus\bigoplus_{j'=1}^{j-1} f_{j'}(C_{j'}(\myModifier{S}))\\
&= f_j(t_m)\oplus g(j-1,\myModifier{t},\myModifier{k}),
\end{align*}	
where the last equality is due to the choice of $\myModifier{S}$ as a schedule with objective value $g(j-1, \myModifier{t}, \myModifier{k})$. Hence, the ``$\leq$'' direction in (\ref{eqn:AlgoStep}) follows due to the choice of $\myModifier{t}$ and $\myModifier{k}$ as minimizers of the right-hand side.\qed
\end{proof}

Using Lemmas \ref{Lem:optSchedGammaI}, \ref{Lem:DynStart}, and \ref{Lem:DynRecurrence} we can now prove correctness of Algorithm \ref{Alg:DynProg} and show that its runtime is bounded polynomially (if $m$ is fixed).
\begin{theorem}\label{Thm:ConstantM}
	Consider a PFB instance with a constant number of $m$ machines and a regular sum or bottleneck objective function. Then, for a given ordering of the jobs, Algorithm \ref{Alg:DynProg} finds the best permutation schedule in time $\mathcal{O}(n^{m^2+4m+1})$.
\end{theorem}
\begin{proof}
	Without loss of generality, assume that the desired job ordering is given by the job indices.
	By Lemmas \ref{Lem:DynStart} and \ref{Lem:DynRecurrence}, Algorithm \ref{Alg:DynProg} correctly computes the values of function $g$. Lemma \ref{Lem:optSchedGammaI} assures that the resulting schedule is indeed optimal among all permutation schedules ordered by the job indices.
	
	Concerning the time complexity, note that there are $\left\lvert \Gamma \right\rvert$ many possible choices for vector $t$ and $\left\lvert\mathcal{B} \right\rvert $ many possible choices for vector $k$, yielding a total of $\left\lvert \Gamma \right\rvert \cdot\left\lvert\mathcal{B} \right\rvert$ possible combinations for vectors $t$ and $k$.
 	The number can be bounded from above as follows:
	\begin{equation*}
	\begin{array}{ccccc}
	\left\lvert \Gamma \right\rvert \cdot \left\lvert\mathcal{B} \right\rvert&
	=\ & \abs{\Gamma_1}\cdot\abs{\Gamma_2}\cdot\ldots\cdot\abs{\Gamma_m} &\cdot& b_1 \cdot b_2 \cdot\ldots\cdot b_m\\
	&\leq\ & n^2 \cdot n^3 \cdot \ldots \cdot n^{m+1} & \cdot & n^m\\
	&=\ & n^{\frac{(m+1)(m+2)}{2}-1+m}& &\\
	&=\ & n^{\frac{m^2}{2}+\frac{5m}{2}}.& &
	\end{array}
	\end{equation*}
	Here, the first $m$ factors in the second line are due to $\Gamma_i \leq n^{i+1}$ and the last factor $n^m$ is due to inequalities $b_i \leq n$ for all $i \in \left[m\right]$.
	
	Therefore, Step \ref{Line:Start} of Algorithm \ref{Alg:DynProg} involves at most $n^{\frac{m^2}{2}+\frac{5m}{2}}$ iterations, each of which can be performed in constant time because $m$ is fixed.
	
	For Step \ref{Line:Recurrence}, first observe that conditions $(*)$ and $(**)$ can be checked in constant time for fixed $m$. In a single iteration, that is, for fixed $j$, $t$, and $k$, in order to compute the minimum in the recurrence of Lemma \ref{Lem:DynRecurrence}, we may have to consider all possible choices for $t'$ and $k'$. The number of these choices is $\abs{\{(t',k')\in\Gamma\times\mathcal{B}\mid(**)\}}$. Aggregating over all $t$ and $k$, while keeping $j$ fixed, yields
	\begin{align*}
	&\quad\sum_{(t,k)\in\Gamma\times\mathcal{B}}\abs{\{(t',k')\in\Gamma\times\mathcal{B}\mid(**)\}}\\
	&= \abs{\{(t,t',k,k')\in\Gamma \times \Gamma \times\mathcal{B} \times \mathcal{B}\mid(**)\}}\\
	&= \prod_{i=1}^m\ \abs{\{(t_i,t'_i,k_i,k'_i) \in\Gamma_i \times \Gamma_i \times[b_i] \times [b_i]
		 \\ 
		& \qquad \qquad \qquad \qquad \qquad \qquad \qquad \qquad \mid(v)\text{ and }(vi)\}}.
	\end{align*}
	The number of quadruples $(t,t',k,k')$ satisfying (v) and (vi) can be bounded by $b_i\abs{\Gamma_i}^2+b_i\abs{\Gamma_i}$. Here, the first term corresponds to the quadruples with $k_i=1$, while the second term corresponds to the quadruples with $k_i>1$, in which case $t'_i$ and $k'_i$ are uniquely determined by $t_i$ and $k_i$.
	Further, note that due to $\abs{\Gamma_i} \geq 1$, it holds that $b_i\abs{\Gamma_i}^2+b_i\abs{\Gamma_i} \leq 2 b_i \abs{\Gamma_i}^2$. Hence, the total running time of Step \ref{Line:Recurrence} (also aggregating over all choices of $j$), is at most
	\begin{align}
	\begin{split}\label{Equ:RuntimeEstimate}
	&\phantom{{}={}}\mathcal{O}\left(n\prod_{i=1}^m 2 b_i \abs{\Gamma_i}^2\right)\\
	&=\mathcal{O}\left(n \prod_{i=1}^m b_i \abs{\Gamma_i}^2\right)\\	&\subseteq\mathcal{O}\left(n\prod_{i=1}^m n\cdot n^{2i+2}\right)\\
	&=\mathcal{O}(n^{m^2+4m+1}).
	\end{split}
	\end{align}
	The first equality is due to the factor $2^m$ being constant for constant $m$ and the set inclusion is due to $b_i \leq n$ and $\abs{\Gamma_i} \leq n^{i+1}$.
	
	Finally, in Step \ref{Line:Backtrack}, we take the minimum over at most $n^{\frac{m^2}{2}+\frac{5m}{2}}$ many values. The backtracking and reconstruction of the final schedule can be done in time $\mathcal{O}(n)$ for constant $m$.
	
	Hence, in total, Step \ref{Line:Recurrence} is the bottleneck and the total time complexity is $\mathcal{O}(n^{m^2+4m+1})$.\qed
\end{proof}

Combining Theorems \ref{Thm:PermI} and \ref{Thm:ConstantM}, we obtain:

\begin{corollary}\label{Cor:FixedM}
	For a PFB with release dates and a fixed number $m$ of machines, the makespan can be minimized in $\mathcal{O}(n^{m^2+4m+1})$ time. The same holds for the total completion time.\qed
\end{corollary}

In particular, for any fixed number $m$ of machines, the problems
$$Fm \mid r_j, p_{ij} = p_i, p\text{-batch}, b_i \mid f$$ 
with $f\in\{C_{\max}, \sum C_j \}$ can be solved in polynomial time.

\subsection{Improved running time for makespan minimization}

\citet[Theorem 6]{SungEtAl:ProblemReduction} showed that for minimizing the make\-span in a PFB, it is optimal to use the so-called first-only-empty batching on the last machine $M_m$, that is, to completely fill all batches on the last machine, except for the first batch, which might contain less jobs if $n$ is not a multiple of $b_m$. In a permutation schedule ordered by the job indices with first-only-empty batching on $M_m$, a makespan of $c$ can be achieved if and only if for each job $J_j$ we have enough time to process the $n-j+1$ jobs $J_j, J_{j+1},\dots,J_n$ on $M_m$ after $J_j$ has been completed on $M_{m-1}$. In other words, a makespan of $c$ can be achieved, if and only if $c\geq c_{(m-1)j}+\left\lceil\frac{n-j+1}{b_m}\right\rceil$ for every $j\in[n]$. Hence, the optimal makespan is given by
$\max_{j=1,2,\ldots,n} c_{(m-1)j}+\left\lceil\frac{n-j+1}{b_m}\right\rceil$. Therefore, the problem to minimize the makespan on an $m$-machine PFB can be solved in the following way: We first find a schedule for the first $m-1$ machines that minimizes the objective function $\bigoplus_{j=1}^n f_j(c_{(m-1)j})$ with $\bigoplus=\max$ and $f_j(c_{(m-1)j})=c_{(m-1)j}+\left\lceil\frac{n-j+1}{b_m}\right\rceil$. By Theorem \ref{Thm:ConstantM}, this can be done in time \[\mathcal{O}(n^{(m-1)^2+4(m-1)+1})=\mathcal{O}(n^{m^2+2m-2})\]
for fixed $m$.
Afterwards, this schedule is extended by a first-only-empty batching on $M_m$, which does not further increase the asymptotic runtime. Thus, for the makespan objective, we can strengthen Corollary \ref{Cor:FixedM}:

\begin{corollary}\label{Cor:Makespan}
	For a PFB with release dates and a fixed number $m$ of machines, the makespan can be minimized in $\mathcal{O}(n^{m^2+2m-2})$ time. \qed
\end{corollary}

\section{Equal Release Dates} \label{Sec:equalReleaseDates}

The dynamic program presented in the previous section works for a very general class of objective functions. Still, Corollary \ref{Cor:FixedM} only holds for makespan and total completion time because for other standard objective functions permutation schedules are not optimal, as shown in Example \ref{Exa:NonPermSchedule}.

However, in the case where all release dates are equal, it turns out that permutation schedules are always optimal. We first show that in the case of equal release dates, any feasible schedule $S$ can be transformed into a feasible permutation schedule $\SchSec$ with equal objective value.

\begin{lemma} \label{lem:permuWithoutRelease}
	Let $S$ be a feasible schedule for a PFB where all jobs are released at time $0$. Then there exists a feasible permutation schedule $\SchSec$ such that $C_j(S) = C_j(\SchSec)$ for all $j = 1, 2, \ldots, n$.
\end{lemma}
\begin{proof}
	Let $\sigma$ be the ordering of jobs on the last machine.
	Since $r_j = 0$ for all $j = 1, 2, \ldots, n$, ordering $\sigma$ is an earliest release date ordering.
	Thus we can use Lemma~\ref{lem:releaseDatePermu} to construct a permutation schedule $\SchSec$, where jobs are ordered by $\sigma$ on all machines and the multi-set of completion times is the same as for $S$.
	Furthermore, since the ordering on the last machine in $\SchSec$ is the same as in $S$, in $\SchSec$ each job has the same completion time as in $S$, i.e.\ $C_j(S) = C_j(\SchSec)$ for all $j=1,2,\ldots,n$. \qed
\end{proof}

From Lemma \ref{lem:permuWithoutRelease}, it follows that permutation schedules are optimal for PFBs with equal release dates, as stated in the following theorem. The theorem generalizes the result presented by \cite{Rachamadugu-1982-15120} for usual proportionate flow shop without batching. Also compare \citep[Lemma 1]{SungEtAl:ProblemReduction}.
In the theorem, we use the term \emph{scheduling objective function} to describe any objective function only dependent on the completion times of the jobs.
\begin{theorem}\label{Thm:PermNoRelease}
	For any scheduling objective function, if an optimal schedule exists for a PFB problem without release dates, then there also exists an optimal permutation schedule. 
\end{theorem}
\begin{proof}
		Apply Lemma \ref{lem:permuWithoutRelease} to an optimal schedule $S^*$ in order to obtain a permutation schedule $\SchSec$. Note that, by Lemma \ref{lem:permuWithoutRelease}, for every job, the completion time on the last machine in schedule $\SchSec$ is exactly the same as in schedule $S^*$. Thus, for any completion time dependent objective function, schedules $S^*$ and $\SchSec$ have the same objective value. \qed
\end{proof}

Note that for all regular objective functions (including those studied in this paper), at least one optimal solution always exists. Indeed, Lemma \ref{Lem:BatchActive} shows that for regular objective functions batch-active schedules are optimal, and there are only finitely many batch active schedules for a PFB.
Thus, the restriction in the statement of the theorem is only a formality.
However, Theorem \ref{Thm:PermNoRelease} holds for a much wider range of scheduling objective functions than studied in this paper. Importantly, since finish times for all jobs stay exactly the same when moving from schedule $S$ to schedule $\SchSec$ in the proof of Theorem \ref{Thm:PermNoRelease}, the theorem holds also for non-regular objective functions (as long as they are defined in such a way that at least one optimal solution exists).

Of course, it may still be hard to find optimal permutations. However, the following theorem shows that for several traditional objective functions, an optimal permutation can be found efficiently.
For the solution of due date related objectives, earliest due date orderings are of particular importance. They are defined analogously to earliest release date orderings: an ordering $\sigma$ is an \emph{earliest due date ordering}, if $d_{\sigma(1)} \leq d_{\sigma(2)} \leq \ldots \leq d_{\sigma(n)}$.

\begin{theorem} \label{th:optimalOrderings}
Consider a PFB without release dates.
	\begin{enumerate}[leftmargin=0.75cm, label=(\roman{enumi})]
	\item For minimizing the total weighted completion time, any ordering by non-increasing weights is optimal.
	\item For minimizing maximum lateness and total tardiness, any earliest due date ordering is optimal.
\end{enumerate}

\end{theorem}
\begin{proof}
	For proving (i), assume jobs are indexed by any given non-increasing order of weights, i.e., $w_{1}\geq w_{2} \geq \dots \geq w_{n}$. By Theorem \ref{Thm:PermNoRelease}, there exists an optimal permutation schedule $S^*$ minimizing the total weighted completion time.
	Suppose there is at least one pair of jobs $J_{j_1}, J_{j_2}$ with $j_1 < j_2$, but $C_{j_1}(S^*) > C_{j_2}(S^*)$.
	Swapping $J_{j_1}$ and $J_{j_2}$ in $S^*$ yields a new permutation schedule $\SchSec$, with objective value
	\begin{eqnarray*}
		\sum w_j C_j (\SchSec) &=&  \sum w_j C_j (S^*) \\
		& & +(w_{j_1}- w_{j_2})(C_{j_2}(S^*) - C_{j_1}(S^*)).
	\end{eqnarray*}
	Due to $j_1 < j_2$, we have $w_{j_1} \geq w_{j_2}$ and with $C_{j_1}(S^*) > C_{j_2}(S^*)$ it follows that \[(w_{j_1}- w_{j_2})(C_{j_2}(S^*) - C_{j_1}(S^*))\leq 0.\]
	Therefore, it we obtain \[\sum w_j C_j (\SchSec) \leq \sum w_j C_j (S^*).\]
	
	Performing such exchanges of jobs sequentially, we eventually obtain an optimal permutation schedule in which jobs are scheduled in the given order of non-increasing weights. Hence, any given ordering of the jobs by non-increasing weights is optimal.
	
	Part (ii) can be proven by exchange arguments analogous to (i). A detailed proof can be found in the appendix.
	\qed
\end{proof}

From Theorems \ref{Thm:ConstantM} and \ref{th:optimalOrderings}, we obtain the following complexity results.
\begin{theorem}\label{Thm:NoRelease}
	For a PFB without release dates and a fixed number $m$ of machines, the total weighted completion time can be minimized in $O(n^{m^2+2m-2})$ time. The same holds for the maximum lateness and the total tardiness.
\end{theorem}
\begin{proof}
	 With a time complexity of $O(n^{m^2+4m+1})$, this statement follows directly from Theorems~\ref{Thm:ConstantM} and~\ref{th:optimalOrderings}. The improvement in the running time can be obtained by the following two observations. First note that, due to vanishing release dates, the size of $\Gamma_i$ is at most $n^i$ instead of $n^{i+1}$. Comparing with the estimate (\ref{Equ:RuntimeEstimate}) in the proof of Theorem \ref{Thm:ConstantM}, this yields an improvement by the factor $n^{2m}$ because each of the $m$ factors $\abs{\Gamma_i}^2$ can be bounded by $n^{2i}$ instead of $n^{2i+2}$. Second, \citet[Theorem 1]{SungEtAl:ProblemReduction} showed that it is optimal to use the so-called last-only-empty batching on the first machine $M_1$, that is, to completely fill all batches on $M_1$, except for the last batch, which might contain less jobs if $n$ is not a multiple of $b_1$. If we do so, the schedule of $M_1$ is fixed. This also fixes all values of the parameters $k_1$ and $t_1$ dependent on the current job index $j$. Hence, in the runtime estimation (\ref{Equ:RuntimeEstimate}), we may leave out the factor $b_1\abs{\Gamma_1}^2$. Since both $\Gamma_1$ and $b_1$ are bounded by $n$ (already considering the improvement above), this results in a further improvement by a factor of $n^3$ in the total runtime. Therefore, we obtain a total time complexity of \[O(n^{m^2+4m+1-2m-3})=O(n^{m^2+2m-2}).\tag*{\qed}\]
\end{proof}

It is also possible to minimize the (weighted) number of late jobs in $O(n^{m^2+2m-1})$ time. However, to achieve this result, it is necessary to adjust the algorithm from Section \ref{Sec:Algorithm} slightly.
Suppose that jobs are ordered in an earliest due date order and suppose further $\mathcal{J}$ is the set of on-time jobs in an optimal solution for PFB to minimize the (weighted) number of late jobs.
Note that in this case, we can find an optimal schedule by first scheduling all jobs in $\mathcal{J}$, in earliest due date order, using the algorithm from Section \ref{Sec:Algorithm}, and then schedule all late jobs in an arbitrary, feasible way.

Usually, however, the optimal set of on-time jobs is not known in advance. Thus the dynamic program has to be adapted in such a way that it finds the optimal set of on-time jobs and a corresponding optimal schedule simultaneously. 
This can be done by, roughly speaking, allowing the algorithm to ignore a job:
in the $j$-th step of the algorithm, when job $J_j$ (in earliest due date order) is scheduled, the algorithm can not only decide to schedule the job in sequence, but can alternatively decide to skip it, making it late (and paying a penalty $w_j$).
In all other regards, the algorithm remains exactly the same as presented in Section \ref{Sec:Algorithm}.
The technical details to prove correctness and running time of the algorithm are, as one would expect, very similar to Section \ref{Sec:Algorithm} and are moved to the appendix.

We obtain the following theorem. The additional factor $n$ in comparison to Theorem \ref{Thm:NoRelease} stems from the fact that we do not know the optimal permutation in advance, and therefore we cannot make use of last-only-empty batching on the first machine as easily as in the proof of Theorem \ref{Thm:NoRelease}.
\begin{theorem}\label{Thm:wjUj}
	For a PFB without release dates and a fixed number $m$ of machines, the weighted number of late jobs can be minimized in $\mathcal{O}(n^{m^2+2m-1})$ time.
\end{theorem}

In particular, combining Theorem \ref{Thm:NoRelease} and Theorem \ref{Thm:wjUj}, we obtain that, for any fixed number $m$ of machines, the problems
$$Fm \mid p_{ij} = p_i, p\text{-batch}, b_i \mid f$$ 
with $f\in\{\sum w_jC_j, L_{\max}, \sum T_j, \sum U_j, \sum w_j U_j \}$ can be solved in polynomial time.

Note that, while minimizing makespan and total completion time in a PFB with a fixed number $m$ of machines and without release dates can be handled in the same manner, a time complexity of 
$O(n^{m^2+2m-2})$ is strictly speaking no longer polynomial for those two problems.
Indeed, since in problems $$Fm\mid p_{ij} = p_i, p\text{-batch}, b_i\mid f$$ with $f \in \{C_{\max}, \sum C_j \}$,
jobs are completely identical, an instance can be encoded concisely by simply encoding the number $n$ of jobs, rather than every job on its own.
Such an encoding of the jobs takes only $O(\log n)$ space, meaning that even algorithms linear in $n$ would no longer be polynomial in the length of the instance.

In the literature, scheduling problems for which such a concise encoding in $O(\log n)$ space is possible are usually referred to as \emph{high-multiplicity scheduling problems} \citep{HochbaumShamir:HighMultI, HochbaumShamir:HighMultII}. They are often difficult to handle and require additional algorithmic consideration, since even specifying a full schedule normally takes at least $O(mn)$ time (one start and end time for each job on each machine).
In order to resolve this issue, for these high-multiplicity scheduling problems special attention is often paid to finding optimal schedules with specific, regular structures. Such a structure may then allow for a concise $O(\log n)$ encoding of the solution. Notice, however, that a priori the dynamic programming algorithm presented in this paper does not give rise to such a regular structure.

In the case of $m =2$ machines and the makespan objective, \cite{Hertrich2018} shows that the $O(n)$ algorithm by \cite{Ahmadi:Batching} can be adjusted to run in polynomial time, even in the high multiplicity sense. 
For all other cases, the theoretical complexity of the two problems $$Fm\mid p_{ij} = p_i, p\text{-batch}, b_i\mid f$$ with $f \in \{C_{\max}, \sum C_j \}$ is left open for future research.
However, note that in practical applications jobs usually have additional meaning attached (for instance patient identifiers, in our pharmacy example from the beginning) and thus such instances are usually encoded with a length of at least $n$.

\section{Conclusions and Future Work}\label{Sec:Conclusion}

In this paper, motivated by an application in modern pharmaceutical production, we investigated the complexity of scheduling a proportionate flow shop of batching machines (PFB). 
The focus was to provide the first complexity results and polynomial time algorithms for PFBs with $m> 2$ machines. 

Our main result is the construction of a new algorithm, using a dynamic programming approach, which schedules a PFB with release dates to minimize the makespan or total completion time in polynomial time for any fixed number $m$ of machines.
In addition, we showed that, if all release dates are equal, the constructed algorithm can also be used to minimize the weighted total completion time, the maximum lateness, the total tardiness, and the (weighted) number of late jobs. Hence, for a PFB without release dates and with a fixed number of $m$ machines, these objective functions can be minimized in polynomial time. Previously these complexity results were only known for the special case of $m=2$ machines, while for each $m\geq 3$ the complexity status was unknown. For minimizing weighted total completion time and weighted number of late jobs, these results were even unknown in the case of $m=2$.

An important structural result in this paper is that permutation schedules are optimal for PFBs with release dates and the makespan or total completion time objective, as well as for any PFB problem without release dates.
This result was needed in order to show that the constructed algorithm can be correctly applied to the problems named above.

Concerning PFBs with release dates, recall that in the presence of release dates permutation schedules are not necessarily optimal for traditional scheduling objective functions other than makespan and total completion time (see Example \ref{Exa:NonPermSchedule}).
However, there are several special cases for which permutation schedules remain optimal, in particular if the order of release dates corresponds well to the order of due dates (in the case of due date objectives) and/or weights (in the case of weighted objectives). 
For example, \cite{Hertrich2018} shows that, when minimizing the weighted total tardiness, if there exists a job order $\sigma$ that is an earliest release date ordering, an earliest due date ordering, and a non-increasing weight ordering at the same time, then there exists an optimal schedule which is a permutation schedule with jobs ordered by $\sigma$.
This, of course, implies analogous results for weighted total completion time, total tardiness and maximum lateness (again, see \cite{Hertrich2018}).
In those special cases, clearly the algorithm presented in this paper can be used to solve the problems in polynomial time.
Still, the complexity status of PFBs with release dates in the general case, where orderings do not correspond well with each other, remains open for all traditional scheduling objectives other than makespan and total completion time.

Another problem left open for future study is the complexity of minimizing the weighted total tardiness in the case where all release dates are equal. This problem is open, even for the special case of $m = 2$.
In this paper, it is shown that permutation schedules are optimal for all objective functions in the absence of release dates. Furthermore, if an optimal job order is known and the number of machines is fixed, then an optimal schedule can be found in polynomial time, using the algorithm presented in this paper.

Unfortunately, the complexity of finding an optimal job order to minimize the weighted total tardiness in a PFB remains open.
Note that for scheduling a single batching machine with equal processing times and for usual proportionate flow shop without batching the weighted total tardiness can be minimized in polynomial time (\cite{SchedulingComplexityWebsite}). Here, proportionate flow shop can be reduced to the single machine problem with all processing times equal to the processing time of the bottleneck machine (see, e.g., \cite{SungEtAl:ProblemReduction}). 
However, such a reduction is not possible for PFBs, since there may be several local bottleneck machines which influence the quality of a solution.

Observe also that, in the special case of $m = 2$ machines, given an arbitrary schedule for the first machine, finding an optimal schedule for the second machine is similar to solving problem 
$1\mid p_j = p, r_j \mid \sum w_j T_j$ or the analogous problem on a batching machine, which also both have open complexity status (see \cite{Baptiste2000:BatchingIdenticalJobs, SchedulingComplexityWebsite}).
So attempting to split up the problem in such a manner does not work either. 
Of course, if instead of an arbitrary schedule on the first machine, some special schedule is selected, it may be possible that the special structure for the release dates on the second machine helps to solve the problem. Notice, though, that this still leaves open the question of which schedule to select on the first machine.

The last open complexity question concerns the case where the number of machines is no longer fixed, but instead part of the input. We are unfortunately not aware of any promising approaches to find reductions from NP-hard problems. However, interestingly, \cite{Hertrich2018} proves that for all versions of PFBs, minimizing the total completion time is always at least as hard as minimizing the makespan. This reduction does not hold for scheduling problems in general. 

One noteworthy special case, where positive complexity may be more readily achievable, arises from fixing the number of jobs $n$, leaving only the number of machines $m$ as part of the input. 
Note that in this case, high-multiplicity, as described in Section \ref{Sec:equalReleaseDates}, is not a concern: since $n$ is fixed, a schedule can be put out in $O(m)$ time and since each machine has an individual processing time, the input also has at least size $O(m)$.
For PFBs with exactly two jobs, $n=2$, and an arbitrary number of machines, \cite{Hertrich2018} shows that the makespan can be minimized in $O(m^2)$ time.
A machine-wise dynamic program is provided that finds all schedules, in which the completion time of one job cannot be improved without delaying the other job. In other words, the program constructs all Pareto-optimal schedules, if the completion time of each job is considered as a separate objective. 
If the number of jobs $n$ is larger than $2$, then the dynamic program can still be applied, but its runtime in that case is pseudo-polynomial (see \cite{Hertrich2018}).

Finally, note that while the algorithm provided in this paper has the benefit of being very general, thus being applicable to many different problems, it incurs some practical limitations due to its relatively large running time. For example, for $m=2$, our algorithm to minimize the makespan has a runtime of $\mathcal{O}(n^{6})$ (Corollary \ref{Cor:Makespan}). In contrast, the algorithm specialized to the case of $m=2$ by \citet{SungYoon:DynamicArrivalsTwoBPM} runs in $\mathcal{O}(n^2)$. The reason for the larger runtime of our dynamic program is that, due to its generality, it needs to consider all possible job completion times in the set $\Gamma$, while the algorithm by \citet{SungYoon:DynamicArrivalsTwoBPM} makes use of the fact that this is not necessary in the special case $m=2$.
 
Now that many of the previously open complexity questions are solved, for future research it would be interesting to find more practically efficient algorithms, in particular for small numbers of machines like $m = 3$ or $m = 4$, and/or for specific objective functions. 
This might be possible by using some easy structural properties, which were not of use for speeding up the algorithm in the general case but which may well help with dedicated algorithms for small numbers of machines. For example, if the objective function is regular, then on any machine, when a batch is started, it is always optimal to include as many jobs as possible (either all waiting jobs, or as many as the maximum batch size).
For the general algorithm provided in this paper, this property was not of use, as our algorithm adds the jobs one by one to the schedule and it only becomes clear how many jobs wait at a machine during the backtracking phase.
For a dedicated algorithm, on the other hand, it may well be possible use this property to reduce the number of possible schedules that have to be considered.

If fast, dedicated algorithms cannot be obtained, then instead approximations or heuristics need to be considered in order to solve PFBs in practice. In the literature, approximations and heuristics often use restrictions to permutation schedules in order to simplify difficult scheduling problems. 
To this end, the results on the optimality of permutation schedules presented in this paper can be of great use to future work. In particular, the results from Section \ref{Sec:Perm} can be of help when designing fast heuristics for makespan, total completion time and related objective functions.
For the same reason, it would be interesting to quantify the quality of permutation schedules (e.g. in terms of an approximation factor) for PFB problems where permutation schedules are not optimal. 

Another approach to improving the algorithms' efficiency, especially if negative complexity results for the case with an arbitrary number of machines are achieved, would be to consider PFBs from a parameterized complexity point of view. Recently, parameterized complexity has started to receive more attention in the scheduling community (cf.  \cite{Mnich2015, WeissEtAl2016, Mnich2018, DiesselAckermann2019}).
A problem with input size $n$ is said to be \emph{fixed-parameter tractable} with respect to a parameter $k$ if it can be solved in a running time of $\mathcal{O}(f(k)p(n))$ where $f$ is an arbitrary computable function and $p$ a polynomial not depending on $k$.
Although our algorithm is polynomial for any fixed number $m$ of machines, it is not fixed-parameter tractable in $m$ because $m$ appears in the exponent of the running time. 
Clearly, a fixed-parameter tractable algorithm would be preferable, both in theory and, most likely, in practice, if one could be found.

%


\bibliographystyle{spbasic}      
\bibliography{PFB_literature}   

%
%

\section*{Appendix A: Proof of Theorem \ref{th:optimalOrderings} (ii)}
In this appendix we give a detailed proof of Theorem \ref{th:optimalOrderings} (ii). Assume jobs are indexed in any given earliest due date ordering, i.e.\ $d_{1}\leq d_{2}\leq\dots\leq d_{n}$. By Theorem \ref{Thm:PermNoRelease}, there exists an optimal permutation schedule $S^*$ with respect to maximum lateness or total tardiness.
Suppose there is at least one pair of jobs $J_{j_1}$, $J_{j_2}$ with $j_1 < j_2$, but $C_{j_1}(S^*)>C_{j_2}(S^*)$.
Swapping $J_{j_1}$ and $J_{j_2}$ in $S^*$ yields a new permutation schedule $\SchSec$.

We first show that $L_{\max} (\SchSec) \leq L_{\max}(S^*)$. Because the lateness of all jobs other than $J_{j_1}$ and $J_{j_2}$ remains unchanged when performing the swap, it suffices to show that
\[\max \{L_{j_1}(\SchSec), L_{j_2}(\SchSec)\} \leq \max \{L_{j_1}(S^*), L_{j_2}(S^*)\}\]
where $L_j(\SchSec) = C_j(\SchSec) - d_j$.
However, due to $j_1 < j_2$, we have $d_{j_1} \leq d_{j_2}$. Hence, with $C_{j_1}(S^*)>C_{j_2}(S^*)$ it follows that \[L_{j_1}(\SchSec) = C_{j_2}(S^*) - d_{j_1} < C_{j_1}(S^*) - d_{j_1} = L_{j_1}(S^*)\]
and 
\[L_{j_2}(\SchSec) = C_{j_1}(S^*) - d_{j_2} \leq C_{j_1}(S^*) - d_{j_1} = L_{j_1}(S^*),\]
which finishes the proof for maximum lateness.

Finally, we show that $\sum T_j (\SchSec) \leq\sum T_j(S^*)$. Because the tardiness of all jobs other than $J_{j_1}$ and $J_{j_2}$ remains unchanged when performing the swap, it suffices to show that
\[T_{j_1}(\SchSec) + T_{j_2}(\SchSec) \leq T_{j_1}(S^*)+ T_{j_2}(S^*)\]
where $T_j(\SchSec) = \max\{C_j(\SchSec) - d_j,0\}$.
We distinguish three cases. Firstly, suppose $J_{j_1}$ and $J_{j_2}$ are both late in $\SchSec$, i.e., $C_{j_1}(\SchSec)> d_{j_1}$ and $C_{j_2}(\SchSec)> d_{j_2}$. Then it follows that
\begin{align*}
T_{j_1}(\SchSec) + T_{j_2}(\SchSec) 
&= C_{j_1}(\SchSec)-d_{j_1} + C_{j_2}(\SchSec)-d_{j_2}\\
&= C_{j_2}(S^*)-d_{j_1} + C_{j_1}(S^*)-d_{j_2}\\
&= C_{j_1}(S^*)-d_{j_1} + C_{j_2}(S^*)-d_{j_2}\\
&\leq T_{j_1}(S^*) + T_{j_2}(S^*).
\end{align*}
Secondly, suppose $J_{j_1}$ is on time in $\SchSec$, i.e.\ $C_{j_1}(\SchSec)\leq d_{j_1}$. Using $d_{j_1} \leq d_{j_2}$, we obtain
\begin{align*}
T_{j_1}(\SchSec) + T_{j_2}(\SchSec) 
&= 0 + \max\{C_{j_2}(\SchSec)-d_{j_2},0\}\\
&= \max\{C_{j_1}(S^*)-d_{j_2},0\}\\
&\leq \max\{C_{j_1}(S^*)-d_{j_1},0\}\\
&= T_{j_1}(S^*)\\
&\leq T_{j_1}(S^*) + T_{j_2}(S^*).
\end{align*}
Thirdly, suppose $J_{j_2}$ is on time in $\SchSec$, i.e.\ $C_{j_2}(\SchSec)\leq d_{j_2}$. Using $C_{j_1}(S^*)> C_{j_2}(S^*)$, we obtain
\begin{align*}
T_{j_1}(\SchSec) + T_{j_2}(\SchSec) 
&= \max\{C_{j_1}(\SchSec)-d_{j_1},0\}+0\\
&= \max\{C_{j_2}(S^*)-d_{j_1},0\}\\
&\leq \max\{C_{j_1}(S^*)-d_{j_1},0\}\\
&= T_{j_1}(S^*)\\
&\leq T_{j_1}(S^*) + T_{j_2}(S^*).
\end{align*}

Hence, performing the described swap does neither increase the maximum lateness nor the total tardiness. Therefore, as in (i), we can perform such exchanges sequentially in order to obtain an optimal permutation schedule with jobs scheduled in the given earliest due date order. Hence, any given earliest due date ordering of the jobs is optimal. \qed

\section*{Appendix B: Proof of Theorem \ref{Thm:wjUj}}
Suppose that jobs are indexed in an earliest due date order, i.e.\ $d_1\leq d_2 \leq ... \leq d_n$. As in Section \ref{Sec:Algorithm}, the modified dynamic program tries to schedule these jobs one by one, but it is also allowed to leave jobs out that will be late anyway, which then will be scheduled in the end after all on-time jobs. This procedure is justified by the following lemma.

\begin{lemma}\label{Lem:UjOrder}
	Consider a PFB instance $I$ without release dates to minimize the weighted number of late jobs. Suppose jobs are indexed by an earliest due date ordering. Then there exists an optimal permutation schedule in which the on-time jobs are scheduled in the order of their indices.
\end{lemma}
\begin{proof}
	By Theorem \ref{Thm:PermNoRelease}, there exists an optimal permutation schedule $S^*$. Let $\mathcal{J}$ be the set of on-time jobs in $S^*$. Let $I_\mathcal{J}$ be the modified PFB instance that contains only the jobs in $\mathcal{J}$. By ignoring all jobs outside of $\mathcal{J}$, schedule $S^*$ defines a feasible permutation schedule $S^*_\mathcal{J}$ for $I_\mathcal{J}$ that has maximum lateness 0. By Theorem \ref{th:optimalOrderings} (ii), there exists also a permutation schedule $S'_\mathcal{J}$ for $I_\mathcal{J}$ in which jobs are scheduled in the order of their indices and that has maximum lateness 0. By scheduling the remaining jobs arbitrarily after the on-time jobs, $S'_\mathcal{J}$ can be extended to a schedule $S'$ for $I$ which has the same weighted number of late jobs as $S^*$ and in which all on-time jobs are scheduled in the order of their indices.\qed
\end{proof}

Recall that for a PFB instance $I$ and a job index $j\in[n]$, $I_j$ denotes the modified instance containing only jobs $J_1,\dots,J_j$. Also, recall the notation $\Gamma=\Gamma_1 \times \Gamma_2 \times \ldots \times \Gamma_m$ and $\mathcal{B}=\left[b_1\right] \times \left[b_2\right] \times \ldots \times \left[b_m\right]$, where $\times$ is the standard Cartesian product. The variables of the new dynamic program are defined almost as in Section \ref{Sec:Algorithm}. However, in contrast, we say that a permutation schedule $S$ for instance $I_j$ \emph{corresponds} to vectors \[t = (t_1, t_2, \ldots, t_i, \ldots, t_{m-1}, t_m) \in \Gamma\]
and
\[k = (k_1, k_2, \ldots, k_i, \ldots, k_{m-1}, k_m) \in \mathcal{B} \]
if the last \emph{on-time} job in $S$ finishes machine $M_i$ at time $t_i$ and is processed in a batch with exactly $k_i$ \emph{on-time} jobs there for all $i\in[m]$. If there are no on-time jobs, then $S$ cannot correspond to any vectors $t$ and $k$. Note that we do not require that $J_j$ itself is on time, just that there is some on-time job. The difference to Section \ref{Sec:Algorithm} is that we ignore all late jobs in this definition. Therefore, only for the on-time jobs it is assumed that they are processed in the order of job indices.

Similar to Section \ref{Sec:Algorithm}, we define $\mySet(j,t,k)$ to be the set of feasible permutation schedules $S$ for instance $I_j$ satisfying the following properties:
\begin{enumerate}
	\item the on-time jobs in $S$ are ordered by their indices,
	\item $S$ corresponds to $t$ and $k$, and
	\item $S$ is $\Gamma$-active.
\end{enumerate}
The variables of our new dynamic program are defined as follows. For $j\in[n]$, $t\in\Gamma$, and $k\in\mathcal{B}$, let
\[w(j,t,k)=w(j,t_1,t_2,\dots,t_m,k_1,k_2,\dots,k_m)\]
be the minimum weighted number of late jobs of a schedule $S\in\mySet(j,t,k)$.
If no such schedule exists, the value of $w$ is defined to be $+\infty$.

Using the above definitions as well as Lemmas \ref{Lem:optSchedGammaI} and~\ref{Lem:UjOrder}, if there is at least one on-time job in the optimal schedule, the minimum weighted number of late jobs is given by $\min_{t,k} w(n,t,k)$, $t\in\Gamma$, $k\in\mathcal{B}$. In Lemmas~\ref{Lem:DynStart2} and \ref{Lem:DynRecurrence2} (below) we provide formulas to recursively compute the values $w(j,t,k)$, $j=1,2,\dots,n$, $t\in\Gamma$, $k\in\mathcal{B}$. In total, we then obtain Algorithm \ref{Alg:DynProgII} as our modified dynamic program.

\begin{algorithm}[!ht]
	\caption{
	}
	\label{Alg:DynProgII}
	\begin{algorithmic}[1]
		\item[\begin{tabular}{ll}
			\textbf{Input:} & A PFB instance without release dates with\\
			&jobs indexed in an earliest due date\\
			&ordering.\\
			\textbf{Output:} & A schedule minimizing the weighted\\
			& number of late jobs.
		\end{tabular}]
		\item[]
		
		\State For each $t\in\Gamma$ and $k\in\mathcal{B}$, compute $w(1,t,k)$ according to Lemma \ref{Lem:DynStart2}.\label{Line:StartII}
		\State For each $j=2,3,\dots,n$, $t\in\Gamma$, and $k\in\mathcal{B}$, compute $w(j,t,k)$ according to Lemma \ref{Lem:DynRecurrence2}. If $w(j,t,k)<+\infty$, also store which job was the last on-time job added to the schedule and its corresponding $t$ and $k$ values.\label{Line:RecurrenceII}
		\State If there are $t$ and $k$ with $w(n,t,k)<+\infty$, find the values of $t$ and $k$ that minimize $w(n,t,k)$ and start a backtracking procedure from them to reconstruct the optimal schedule using the information stored in Step \ref{Line:RecurrenceII}. Otherwise, return an arbitrary feasible schedule.\label{Line:BacktrackII}
	\end{algorithmic}
\end{algorithm}

We now turn to providing and proving the correctness of the recursive formula to compute the values $w(j,t,k)$, $j=1,2,\dots,n$, $t\in\Gamma$, $k\in\mathcal{B}$. As the recursion is slightly more complicated than in Section \ref{Sec:Algorithm}, we first introduce two kinds of auxiliary variables $x$ and $y$ for our dynamic program. These are then used to formulate and prove Lemmas \ref{Lem:DynStart2} and \ref{Lem:DynRecurrence2}, which deal with the starting values $w(1,t,k)$ and the actual recurrence relation, respectively, in analogy to Lemmas \ref{Lem:DynStart} and \ref{Lem:DynRecurrence}. Finally, by connecting these results, we obtain the correctness and runtime of Algorithm \ref{Alg:DynProgII} as stated in Theorem~\ref{Thm:wjUj}.

For any $j\in[n]$, $t_i\in\Gamma_i$, and $k_i\in[b_i]$, $i\in[m]$, let 
\begin{align}\label{Eq:Defx}
x(j,t,k) = \left\{
\begin{array}{ll}
\sum_{j'=1}^{j-1} w_{j'}, & \text{if (i)--(iv) hold},\\
+\infty,  & \text{otherwise,}
\end{array} \right.
\end{align}
where
\begin{enumerate}[leftmargin=0.75cm, label=(\roman{enumi})]
	\item $k_i=1$ for all $i\in[m]$,
	\item $t_1\geq p_1$,
	\item $t_{i+1}\geq t_{i}+p_{i+1}$ for all $i\in[m-1]$, and
	\item $t_m\leq d_j$.
\end{enumerate}
The values of $x$ can be interpreted as the objective value of a schedule in $\mySet(j,t,k)$ for $I_j$ that only schedules $J_j$ on time and all other jobs late.

For $j>1$, $t_i\in\Gamma_i$, and $k_i\in[b_i]$, $i\in[m]$, let
\begin{align}\label{Eq:Defy}
\begin{split}
&y(j,t,k) \\
&\quad = \left\{
\begin{array}{ll}
\min\{w(j-1,t',k') \mid (**) \}, & \text{if (ii)--(iv) hold},\\
+\infty,  & \text{otherwise,}
\end{array} \right.
\end{split}
\end{align}
where the minimum over the empty set is defined to be $+\infty$. Here, $(**)$ is given by conditions 
\begin{enumerate}[leftmargin=0.9cm, label=(\roman{enumi}), resume]
	\item $t'_i\in\Gamma_i$,
	\item $k'_i\in[b_i]$,
	\item if $k_i=1$, then $t'_i\leq t_i-p_i$, and
	\item if $k_i>1$, then $t'_i=t_i$ and $k'_i=k_i-1$,
\end{enumerate}
for all $i\in[m]$. The values of $y$ can be interpreted as the best possible objective value of a schedule in $\mySet(j,t,k)$ that schedules $J_j$ and at least one other job on time.

Analogous to Lemma \ref{Lem:DynStart}, we obtain the following lemma.

\begin{lemma}\label{Lem:DynStart2}
	For $t_i\in\Gamma_i$ and $k_i\in[b_i]$, $i\in[m]$, we have $w(1,t,k)=x(1,t,k)$.
\end{lemma}
\begin{proof}
	As in the proof of Lemma \ref{Lem:DynStart}, conditions \mbox{(i)--(iii)} are necessary for the existence of a schedule in $\mySet(1,t,k)$. Also, since a schedule for $I_1$ can only correspond to $t$ and $k$ if $J_1$ is on time, (iv) is necessary as well. Conversely, if (i)--(iv) are satisfied, then the vector $t$ of completion times uniquely defines a feasible schedule $S \in \mySet(1,t,k)$ by $c_{i1}(S)=t_i$. Note that $S$ is the only schedule in $\mySet(1,t,k)$ and has objective value 0 and therefore, clearly, is optimal.\qed
\end{proof}

Analogous to Lemma \ref{Lem:DynRecurrence}, we now turn to the recurrence formula to calculate $w$ for $j>1$ from the values of $w$ for $j-1$.
\begin{lemma}\label{Lem:DynRecurrence2}
	For $j\geq 2$, $t_i\in\Gamma_i$, and $k_i\in[b_i]$, $i\in[m]$, the values of $w$ are given by
	\begin{align*}
	w(j,t,k) = \min\{x(j,t,k), y(j,t,k), w(j-1,t,k)+w_j\}.
	\end{align*}
\end{lemma}
Roughly speaking, for scheduling $J_j$ there are three possibilities: either job $J_j$ is scheduled as the only on-time job, as one of several on-time jobs, or as a late job. In the first two cases, $w(j,t,k)$ equals $x(j,t,k)$ or $y(j,t,k)$, respectively, while in the third case it equals $w(j-1,t,k)+w_j$, as we will see in the proof.

{\renewcommand*{\proofname}{Proof of Lemma \ref{Lem:DynRecurrence2}:}
\begin{proof}
	We first prove the ``$\geq$'' direction. If the left-hand side equals $+\infty$, then this direction follows immediately. Otherwise, by definition of $w$, there exists a schedule $S$ for $I_j$ in $\mySet(j,t,k)$ with objective value $w(j,t,k)$. Schedule $S$ naturally defines a feasible permutation schedule $S'$ for instance $I_{j-1}$ by ignoring job $J_j$. We distinguish three cases.
	
	Firstly, if $J_j$ is not on time in $S$, then $S'$ is in $\mySet(j-1,t,k)$ and has an objective value of $w(j,t,k)-w_j$ because $J_j$ is not contained in instance $J_{j-1}$. Therefore, by definition of $w(j-1,t,k)$, it follows that $w(j-1,t,k)\leq w(j,t,k)-w_j$, which implies the ``$\geq$'' direction in the first case.
	
	Secondly, if $J_j$ is the only on-time job in $S$, then (i) must hold because $S$ corresponds to $k$ and only on-time jobs are counted in this definition. Furthermore, (ii) and (iii) must hold because $S$ is feasible and corresponds to $t$. Finally (iv) holds because $J_j$ is on time. Hence, we obtain $x(j,t,k) = \sum_{j'=1}^{j-1} w_{j'} = w(j,t,k)$ in this case, which implies the ``$\geq$'' direction in the second case.
	
	Thirdly, if $J_j$ is one of at least two on-time jobs in $S$, then $S'$ corresponds to unique vectors $t'$ and $k'$. Hence, we obtain $S'\in\mySet(j-1,t',k')$ and, therefore,
	\begin{align*}
	w(j,t,k) &= \sum_{j'=1}^j w_{j'} U_{j'}(S) \\&= \sum_{j'=1}^{j-1} w_{j'} U_{j'}(S')\\&\geq w(j-1,t',k').
	\end{align*}
	Analogous arguments to the proof of Lemma \ref{Lem:DynRecurrence} show that $t'$ and $k'$ satisfy $(**)$. Moreover, as in the second case, (ii)--(iv) hold because $S$ is feasible, corresponds to $t$, and $J_j$ is on time. Therefore, we have
	\begin{align*}
	w(j,t,k) \geq w(j-1,t',k') \geq y(j,t,k),
	\end{align*}
	which completes the ``$\geq$'' direction.
	
	Now we prove the ``$\leq$'' direction. If the right-hand side equals $+\infty$, then this direction follows immediately. Otherwise, we again distinguish three cases.
	
	Firstly, suppose the minimum on the right-hand side is attained by the term $w(j-1,t,k)+w_j$. Since this must be finite, there is a schedule $S'\in\mySet(j-1,t,k)$ for instance $I_{j-1}$ with objective value $w(j-1,t,k)$. By scheduling $J_j$ as a late job, $S'$ can be extended to a schedule $S$ that has objective value $w(j-1,t,k)+w_j$. Moreover, since $J_j$ is late, $S$ still corresponds to $t$ and $k$, i.e.\ $S\in\mySet(j,t,k)$, which implies $w(j,t,k)\leq w(j-1,t,k)+w_j$. Hence, the ``$\leq$'' direction follows in this case.
	
	Secondly, suppose the minimum on the right-hand side is attained by the term $x(j,t,k)$. Since this must be finite, we obtain that (i)--(iv) hold. Hence, the schedule $S$ which schedules $J_j$ such that $c_{ij}=t_i$ and all other jobs late is feasible for $I_j$, an element of $\mySet(j,t,k)$, and has objective value $x(j,t,k)=\sum_{j'=1}^{j-1}w_{j'}$. This implies $w(j,t,k)\leq x(j,t,k)$. Hence, the ``$\leq$'' direction follows in this case.
	
	Thirdly, suppose the minimum on the right-hand side is attained by the term $y(j,t,k)$. Since this must be finite, we obtain that (ii)--(iv) hold. Let $t'$ and $k'$ be the minimizers in the definition of $y(j,t,k)$. Again due to finiteness, there must be a schedule $S'\in\mySet(j-1,t',k')$ for $I_{j-1}$ with objective value $w(j-1,t',k')=y(j,t,k)$. By ignoring late jobs in $S'$ (which can be scheduled arbitrary late), we can use the same arguments as in the proof of Lemma \ref{Lem:DynRecurrence} to extend $S'$ to schedule $S\in\mySet(j,t,k)$ for $I_j$. Due to (iv), $J_j$ is on time in $S$. Hence, $S$ has the same objective value as $S'$, namely $y(j,t,k)$. Thus, it follows that $w(j,t,k)\leq y(j,t,k)$, which completes the ``$\leq$'' direction. \qed
\end{proof}}

Analogous to Theorem \ref{Thm:ConstantM}, we now deduce the desired theorem by proving correctness and running time of Algorithm \ref{Alg:DynProgII}.

{\renewcommand*{\mytheoremname}{Theorem \ref{Thm:wjUj}}
	\begin{mytheorem}
		For a PFB without release dates and a fixed number $m$ of machines, the weighted number of late jobs can be minimized in $\mathcal{O}(n^{m^2+2m-1})$ time.
	\end{mytheorem}
	\begin{proof}
	Without loss of generality, jobs are indexed in an earliest due date ordering. By Lemmas \ref{Lem:DynStart2} and \ref{Lem:DynRecurrence2}, Algorithm \ref{Alg:DynProgII} correctly computes the values of function $w$. If there is a feasible schedule where at least one job is on time, we obtain due to Lemmas \ref{Lem:optSchedGammaI} and \ref{Lem:UjOrder} that the minimum weighted number of late jobs is exactly $\min_{t,k} w(n,t,k)$, where the minimum is taken over each $t\in\Gamma$ and $k\in\mathcal{B}$. In this case the backtracking procedure in Step \ref{Line:BacktrackII} retrieves an optimal schedule. If, however, in every feasible schedule all jobs are late, then all values of $w$ will be $+\infty$. In this case, any feasible schedule is optimal, with objective value $\sum_{j=1}^n w_n$. Hence, in any case, the schedule returned by Algorithm \ref{Alg:DynProgII} is optimal.
	
	Recall from the proof of Theorem \ref{Thm:ConstantM} that $\left\lvert \Gamma \right\rvert \cdot \left\lvert\mathcal{B} \right\rvert\leq n^{\frac{m^2}{2}+\frac{5m}{2}}$. Therefore, Step \ref{Line:StartII} of Algorithm \ref{Alg:DynProgII} involves at most $n^{\frac{m^2}{2}+\frac{5m}{2}}$ iterations, each of which can be performed in constant time because $m$ is fixed.

	Similarly, Step \ref{Line:RecurrenceII} involves at most $\left(n-1\right) n^{\frac{m^2}{2}+\frac{5m}{2}} \leq n^{\frac{m^2}{2}+\frac{5m}{2}+1}$ iterations. In each of these iterations, computing $x(j,t,k)$ with (\ref{Eq:Defx}) can be done in constant time since $m$ is fixed. The same holds for computing $w(j,t,k)$ using Lemma \ref{Lem:DynRecurrence2}, once $y(j,t,k)$ is known. However, the computation of $y(j,t,k)$ by applying (\ref{Eq:Defy}) is the bottleneck of this step. In full analogy to the proof of Theorem \ref{Thm:ConstantM}, we obtain that the total time complexity of computing all $y$-values is bounded by $\mathcal{O}\left(n\prod_{i=1}^m b_i \abs{\Gamma_i}^2\right)$. Due to vanishing release dates, we have $\abs{\Gamma_i}\leq n^i$. Comparing with the estimation (\ref{Equ:RuntimeEstimate}) in the proof of Theorem \ref{Thm:ConstantM}, this results in a runtime bound of $\mathcal{O}(n^{m^2+2m+1})$.
	
	However, the following argument allows us to reduce this runtime by another factor of $n^2$, resulting in $\mathcal{O}(n^{m^2+2m-1})$. Even though we do not know the optimal job permutation beforehand, we know by \citep[Theorem 1]{SungEtAl:ProblemReduction} that it is optimal to use the so-called last-only-empty batching on the first machine $M_1$, that is, to completely fill all batches on $M_1$, except for the last batch, which might contain less jobs if $n$ is not a multiple of $b_1$. This has two consequences. First, it suffices to consider the $\left\lceil \frac{n}{b_1} \right\rceil$ possible completion times $\Gamma'_1=\left\{p_1, 2p_1, \dots, \left\lceil \frac{n}{b_1} \right\rceil p_1\right\}$ on $M_1$, instead of all elements of $\Gamma_1$. Second, for given $t_1\in\Gamma'_1$ and $k_1\in[b_1]$, the possible previous values $t'_1$ and $k'_1$ are uniquely determined. Thus, in our runtime estimate, we may replace $\abs{\Gamma_1}$ with $\abs{\Gamma'_1}$ and leave out the square, such that instead of the term $b_1\abs{\Gamma_1}^2$ we now have $b_1\abs{\Gamma'_1}\in\mathcal{O}\left(b_1\frac{n}{b_1}\right)=\mathcal{O}(n)$. Previously, this term was bounded by $\mathcal{O}(n^3)$, resulting in the claimed improvement.

	Finally, in Step \ref{Line:BacktrackII}, we take the minimum over at most $n^{\frac{m^2}{2}+\frac{5m}{2}}$ many values. The backtracking and reconstruction of the final schedule can be done in time $\mathcal{O}(n)$ for constant $m$.

	Hence, in total, Step \ref{Line:RecurrenceII} is the bottleneck and the total time complexity is $\mathcal{O}(n^{m^2+2m-1})$.\qed
	\end{proof}}

\end{document}